\documentclass[12pt]{article}
\usepackage{xcolor, amsmath, amssymb, amsthm, bm, multirow, url, graphicx, comment, caption, subcaption, enumitem, booktabs}
\usepackage{authblk}
\usepackage{enumerate}
\usepackage[title]{appendix}
\usepackage[authoryear]{natbib}
\usepackage[colorlinks,citecolor=blue,urlcolor=blue,hypertexnames=false]{hyperref}

\usepackage[linesnumbered, ruled, noline]{algorithm2e}
\SetKwInput{KwInit}{Initialization}
\SetKwInput{KwInput}{Input}
\SetKwInput{KwOutput}{Output}

\newcommand{\blind}{0}

\theoremstyle{plain}
\newtheorem{theorem}{Theorem}
\newtheorem{lemma}[theorem]{Lemma}
\newtheorem{prop}[theorem]{Proposition}
\newtheorem{corollary}[theorem]{Corollary}
\newtheorem{remark}[theorem]{Remark}
\newtheorem{model}{DGP}

\begin{document}

\def\spacingset#1{\renewcommand{\baselinestretch}%
{#1}\small\normalsize} \spacingset{1}


\if0\blind
{
  \title{\bf Temporal Wasserstein Imputation: A Versatile Method for Time Series Imputation} 
  \author{Shuo-Chieh Huang 
    \vspace{-0.4cm}
    \\
    Department of Statistics, Rutgers University\vspace{0.2cm} \\
    Tengyuan Liang \\
    Booth School of Business, University of Chicago\vspace{0.2cm} \\
    and\vspace{0.2cm} \\
    Ruey S. Tsay \\
    Booth School of Business, University of Chicago}
  \maketitle
} \fi

\if1\blind
{
  \bigskip
  \bigskip
  \bigskip
  \begin{center}
    {\LARGE\bf Temporal Wasserstein Imputation: Versatile Missing Data Imputation for Time Series}
\end{center}
  \medskip
} \fi

\bigskip
\begin{abstract}
    Missing data can significantly hamper standard time series analysis, yet they occur frequently in applications.
    In this paper, we introduce temporal Wasserstein imputation, a novel method for imputing missing data in time series.
    Unlike most existing techniques, our approach is fully nonparametric, circumventing the need for model specification prior to imputation, making it suitable for empirical applications even with nonlinear dynamics.
    Its principled algorithmic implementation can seamlessly handle univariate or multivariate time series with any non-systematic missing pattern. In addition, the plausible range and side information of the missing entries (such as box constraints) can easily be incorporated. 
    Furthermore, our method mitigates the distributional bias 
    common among many existing approaches, ensuring more reliable downstream statistical analysis using the imputed series.
    We establish the convergence of an alternating minimization algorithm to critical points. 
    We also provide conditions under which the marginal distributions of the underlying time series can be identified.
    Numerical experiments, including extensive simulations covering both linear and nonlinear time series and an analysis on a real-world groundwater dataset, corroborate the practical usefulness of the proposed method.
\end{abstract}

\noindent%
\vfill

\newpage
\spacingset{2} 

\section{Introduction}
Data collection is often subject to various disruptions at either the idiosyncratic or system level, leading to missing values. 
In time series analysis, missing data pose serious challenges to the application of many statistical tools because they typically assume that observations were taken at equally spaced time points.
Indeed, the problem of missing data in time series has received wide attention from numerous fields, including environmental and climate studies \citep{Afrifa-Yamoah2020, Junger2015}, geophysics \citep{shen2015, Schoellhamer2001}, and engineering \citep{Li2019, Racault2014}.

Existing methods to cope with missing data abound, which can be classified into two general categories.
The methods of the first category estimate the missing data before any further analysis.
For example, the \textit{optimal interpolator}, which is the expectation of the missing data conditional on the observed data and a given model, has been extensively studied \citep{Harvey1984, Pourahmadi1989, pena1991, gomez1999, ALONSO2008, pena2021, McElroy2020, McElroy2022, McElroy2022-2}.
Without any model assumption, some nonparametric imputation methods are also available.
For instance, the singular spectrum analysis (SSA, \citealp{Kondrashov2006, shen2015, golyandina2001}), which relies on the spectral analysis of the trajectory matrix, is a nonparametric method to impute missing time series data.
Other nonparametric methods include smoothing with polynomial interpolation or B-splines and moment matching (see \citealp{Carrizosa2013} and \citealp{hastie2009}).
The methods of the second category circumvent the imputation stage and directly learns the model parameters.
For example, by specifying a Gaussian likelihood function, one can use the Kalman filter to carry out estimation. See \citet{Harvey_1990} and \citet{little2019}, among others.


However, the aforementioned methods have some drawbacks.
On the one hand, the required model specification of the optimal interpolator and the likelihood-based approaches 
may be ill-informed when there are many missing values.
Furthermore, since computing the optimal interpolator under a nonlinear model can be challenging either computationally or analytically, simple linear models are often used in practice.
Finally, the imputed series generally has a different distribution compared to the original process, 
which biases downstream statistical analysis especially when the number of missing data is non-negligible.
On the other hand, the nonparametric SSA methods work well with structural time series featuring deterministic trends and seasonality, but are not effective in handling processes described by stochastic models, such as the celebrated ARIMA process. This is particularly so when the percentage of missing values is high.

In this paper, we present the temporal Wasserstein imputation (TWI), a novel approach for time series imputation which addresses the aforementioned limitations. 
As a nonparametric method, TWI requires no model specification prior to imputation, allowing the underlying time series to be nonlinear, such as the threshold autoregressive process (TAR, \citealp{tong1983}).
Hence, TWI is highly versatile and can be applied to a wide range of stochastic time series, which nicely complements the SSA-based methods.
Formulated as an optimization problem, TWI can be implemented efficiently through an alternating minimization algorithm. 
This principled algorithmic approach not only is computationally efficient, but can also be easily generalized to handle multivariate or non-stationary time series (see Section \ref{Sec::Method}).
Additionally, it can seamlessly integrate the side information by restricting the admissible imputations, which can be used to enforce natural constraints such as nonnegativity constraints.
The optimization and statistical underpinnings of TWI are rooted in the well-known optimal transport problem \citep{villani2008, villani2021, Cuturi2013}.
The core idea of TWI is that by minimizing the Wasserstein distance (over all admissible imputations) between the marginal distributions before and after a specified time point, the marginal distributions implied by the imputation shall resemble those of the original series, which facilitates downstream and exploratory analyses. 

Recently, optimal transport techniques have gained traction in imputation problems under different contexts.
In their pioneering work, \citet{Muzellec2020} 
first advocated that ``ideal'' imputations should preserve the distribution of the underlying data.
\citet{Zhao2023} called this idea ``distribution matching,'' and devised a deep invertible function approach to imputation, leveraging the statistical and computational properties of optimal transport (see also \citealp{Wu2023}).
\citet{hur2024} proposed a new method for individual treatment effect imputation based on computational optimal transport.
In the time series context, after the completion of this paper we became aware of a concurrent and independent work by \citet{wang2025}, who incorporated a new Fourier-based discrepancy measure into the framework of \citet{Muzellec2020}.
However, because of batch sampling and the non-convexity of the optimization problem, neither \citet{wang2025} nor \citet{Muzellec2020} has provided satisfactory optimization or statistical guarantees.
In this paper, we propose an alternating minimization algorithm and study its optimization and statistical properties of TWI.
Our result shows that the algorithm monotonically reduces the loss and is guaranteed to converge to a critical point.
Moreover, we show, in a two-state Markov chain example, that the correct underlying distributions can indeed be identified by TWI under certain assumptions. 
This simple example encapsulates the identification problem and provides an original insight into the distribution matching framework.

The rest of the paper is organized as follows.
Section \ref{Sec::Prelim} provides a brief review of the elements of optimal transport.
Section \ref{Sec::Method} maps out the optimization formulation and introduces the computational scheme of TWI in detail.
It also discusses some extensions and comparisons with certain existing methods.
In Section \ref{Sec::theory}, we investigate the optimization and statistical properties of TWI. 
We also propose a simple modification of TWI based on our theoretical results, termed $k$-TWI, which shows further improvement over TWI in our numerical experiments.
In Section \ref{Sec::Simulation}, we evaluate the performance of TWI on synthetic data generated from a wide range of time series models, including nonlinear, multivariate, and nonstationary settings.
Across all processes considered, TWI successfully recovers the underlying dynamics and yields favorable downstream statistics. 
Its advantage is particularly prominent in the presence of nonlinear dynamics. 
In Section \ref{Sec::realdata}, we apply TWI to a real-world groundwater dataset.
We use the imputed series to compute summary statistics and perform autoregressive modeling.
Compared to existing methods, TWI can effectively capture the underlying dynamics, especially when the missing ratio is high. 


\section{Preliminary: Optimal transport} \label{Sec::Prelim}
In this section, we briefly review the optimal transport (OT) problem, which serves as a prerequisite to the rest of the paper. 
For a more comprehensive account on the theory of optimal transport, the readers are referred to \citet{villani2021, villani2008}.

Given two probability measures $\mu$ and $\nu$ defined on $\mathbb{R}^{p}$ and a cost function $L: \mathbb{R}^{p} \times \mathbb{R}^{p} \rightarrow \mathbb{R}_{+}$, the optimal transport cost between $\mu$ and $\nu$ is defined as 
\begin{align*} 
    \mathsf{OTC}(\mu, \nu) = \inf_{\pi \in \mathcal{M}(\mu, \nu)} \int L(\mathbf{x}, \mathbf{y}) d\pi(\mathbf{x}, \mathbf{y}),
\end{align*}
where $\mathcal{M}$ is the set of couplings between $\mu$ and $\nu$, i.e.,~probability measures over $\mathbb{R}^{p} \times \mathbb{R}^{p}$ with marginals $\mu$ and $\nu$. 
The classic interpretation for the minimization problem is finding the most cost-effective way to move a pile of earth distributed as $\mu$ (with volume normalized to one) to fill a pit whose shape is distributed as $\nu$.
It also has a matching or assignment interpretation that has led to fruitful results in the economics literature \citep{galichon2016optimal}.

If $L(\mathbf{x}, \mathbf{y}) = \Vert \mathbf{x} - \mathbf{y} \Vert^{k}$, for some $k \geq 1$, the resulting optimal transport cost is known as the Wasserstein distance of order $k$ (to the $k$-th power), which we denote as 
\begin{align*}
    \mathcal{W}_{k}(\mu, \nu) := \left\{\inf_{\pi \in \mathcal{M}(\mu, \nu)} \int \Vert \mathbf{x} - \mathbf{y} \Vert^{k} d\pi(\mathbf{x}, \mathbf{y}) \right\}^{1/k} = \{ \mathsf{OTC}(\mu, \nu) \}^{1/k}.
\end{align*}
The Wasserstein distance is a metric on the space of probabilities with finite $k$-th moment on $\mathbb{R}^{p}$, and metrizes the weak topology on this space \citep[Chpt. 9]{villani2021}.
In particular, if $\mu$ and $\{\mu_{t}:t=1,2,\ldots\}$ are probability measures on $\mathbb{R}^{p}$ such that $\mathcal{W}_{k}(\mu_{t}, \mu) \rightarrow 0$ as $t \rightarrow \infty$, then $\mu_{t} \Rightarrow \mu$ and $\int \Vert \mathbf{x} \Vert^{k} \mu_{t}(d\mathbf{x}) \rightarrow \int \Vert \mathbf{x} \Vert^{k} \mu(d\mathbf{x})$, where $\Rightarrow$ denotes the weak convergence of probability measures \citep{billingsley1995}. 
Therefore, it can be used as a discrepancy metric between $\mu$ and $\nu$.
In fact, the Wasserstein distance offers several advantages over other discrepancy metrics such as the Kullbeck-Leibler divergence and the Kolmogorov-Smirnov distance. For instance, the Wasserstein distance has a nice dual representation and handles support mismatch easily, making it suitable for machine learning and statistical applications.

If $\mu$ and $\nu$ are discrete measures, for instance $\mu = \sum_{i=1}^{n} p_{i}\delta_{\mathbf{x}_{i}}$ and $\nu = \sum_{j=1}^{m}q_{j}\delta_{\mathbf{y}_{j}}$, where $\{\mathbf{x}_{i}\}_{i=1}^{n}, \{\mathbf{y}_{j}\}_{j=1}^{m}$ are in $\mathbb{R}^{p}$, $\sum_{i=1}^{n} p_{i} = \sum_{j=1}^{m} q_{j} = 1$, $p_{i} > 0$, $q_{j} > 0$ for all $i,j$, and $\delta_{\mathbf{x}}$ denotes the dirac measure at $\mathbf{x}$, then the minimization problem becomes the following linear program.
\begin{align}
    \mathsf{OTC}(\mu, \nu) = \min_{\{\pi_{ij}\}, \pi_{ij} \geq 0}& \sum_{i,j} L(\mathbf{x}_{i}, \mathbf{y}_{j}) \pi_{ij} \notag \\
    \mbox{s.t. } & \sum_{i=1}^{n} \pi_{ij} = q_{j}, \quad j = 1,\ldots,m, \label{Sec2-1}
    & \sum_{j=1}^{m} \pi_{ij} = p_{i}, \quad i = 1,\ldots,n. 
\end{align}
The constraints in \eqref{Sec2-1} 
enforce the joint distribution $\{\pi_{ij}\}$ to be a coupling between $\mu$ and $\nu$, and will henceforth be called the coupling constraints. 
Recent research has made significant progress in the fast computation of the solution $\{\pi_{ij}\}$ through entropic regularization and some algorithms are amenable to large-scale or GPU parallelization; see \citet{Cuturi2013} and \citet{peyre2019}.

\section{Temporal Wasserstein imputation} \label{Sec::Method} 
This section introduces the temporal Wasserstein imputation (TWI) method for time series imputation. 
We first present its optimization formulation, which shows close connection to optimal transport.
Then, we describe an alternating minimization algorithm, and highlight its relation to some existing approaches.  

\subsection{Optimization formulation}
Let $\{\tilde{x}_{t}\}$ be the time series of interest and we observe the data $\{x_{t}:t = 0, 1, \ldots, n-1\}$.
Denote by $\bigstar$ a generic missing value and $\mathcal{M}$ the set of time indices corresponding to the missing observations.
Thus, $x_{t} = \bigstar$ if $t \in \mathcal{M}$ and $x_{t} = \tilde{x}_{t}$ otherwise. 
We discuss below the case that $x_t$ is stationary. The case of unit-root nonstationary series is discussed at the end of the next subsection.
If $\{\tilde{x}_{t}\}$ is stationary, then its $p$-dimensional marginal distribution 
$\mu_{p}$, i.e.~the distributions of $(\tilde{x}_{t}, \tilde{x}_{t-1}, \ldots, \tilde{x}_{t-p+1})$, does not depend on $t$, where $p \in \mathbb{N}$. 
In view of this property, when imputing the missing values in $(x_{0}, \ldots, x_{n-1})$, we seek an imputation $\mathbf{w} = (w_{0}, w_{1}, \ldots, w_{n-1})^{\top}$ whose marginal distribution remains nearly unchanged.
More specifically, we first split the time series into two parts $\{w_{0}, w_{1}, \ldots, w_{n_{1}}\}$ and $\{w_{n_{1}+1}, w_{n_{1}+2}, \ldots, w_{n-1}\}$   
and define the two empirical marginal distributions
\begin{align*}
    \hat{\mu}_{\mathrm{pre}}(\mathbf{w}) =  \sum_{t=p-1}^{n_{1}} \frac{\delta_{(w_{t}, w_{t-1}, \ldots, w_{t-p+1})}}{n_{1}-p+2}, \quad 
    \hat{\mu}_{\mathrm{post}}(\mathbf{w}) = \sum_{t=n_{1}+1}^{n-1} \frac{\delta_{(w_{t}, w_{t-1}, \ldots, w_{t-p+1})}}{n-n_{1}-1},
\end{align*}
where $p - 1 < n_1 < n - p$, and the values of $p$ and $n_1$ used will be discussed later. The discrepancy between the two empirical distributions can then be measured by the optimal transport cost, which is the minimum value of the following linear program.
\begin{equation}
\left.
\begin{aligned} \label{sec2_OTC}
    \mathsf{OTC} (\hat{\mu}_{\mathrm{pre}}(\mathbf{w}), \hat{\mu}_{\mathrm{post}}(\mathbf{w})) = \min_{\mathbf{\Pi} = (\pi_{ij})_{ij}, \pi_{ij} \geq 0} & \sum_{i = p-1}^{n_{1}} \sum_{j=n_{1}+1}^{n-1} \pi_{ij} L(\mathbf{v}_{i}(\mathbf{w}), \mathbf{v}_{j}(\mathbf{w})) \\
    \mbox{s.t.  } & \sum_{i=p-1}^{n_{1}} \pi_{ij} = \frac{1}{n - n_{1} - 1} \quad \mbox{for all } j\\
    & \sum_{j=n_{1}+1}^{n-1} \pi_{ij} = \frac{1}{n_{1} - p + 2} \quad \mbox{for all } i, 
\end{aligned}
\right\}
\end{equation}
where $\mathbf{v}_{i}(\mathbf{w}) := (w_{i}, w_{i-1}, \ldots, w_{i-p+1})^{\top}$, and, for ease of notation, we enumerate the rows and columns of the $(n_{1} - p + 2) \times (n - n_{1} - 1)$ matrix $\mathbf{\Pi}$ by $\{p-1, p, \ldots, n_{1}\}$ and $\{n_{1}+1, \ldots, n-1\}$, respectively.  
An important special case of \eqref{sec2_OTC} is when $L(\mathbf{v}_{i}(\mathbf{w}), \mathbf{v}_{j}(\mathbf{w})) = \Vert \mathbf{v}_{i}(\mathbf{w}) - \mathbf{v}_{j}(\mathbf{w})\Vert^{k}$, for some $k \geq 1$.
In this case, we have $\mathsf{OTC}(\hat{\mu}_{\mathrm{pre}}, \hat{\mu}_{\mathrm{post}}) = \mathcal{W}_{k}^{k}(\hat{\mu}_{\mathrm{pre}}, \hat{\mu}_{\mathrm{post}})$.

Our goal is to find the imputation that minimizes the aforementioned discrepancy. Thus, we solve for
\begin{align} \label{sec2_what}
    \hat{\mathbf{w}} := (\hat{w}_{0}, \ldots \hat{w}_{n-1})^{\top} = \arg\min_{\mathbf{w} \in \mathcal{C}} \mathsf{OTC}(\hat{\mu}_{\mathrm{pre}}(\mathbf{w}), \hat{\mu}_{\mathrm{post}}(\mathbf{w})),
\end{align}
where the constraint set $\mathcal{C}$ delineates the admissible imputations. 
Typically, it ensures the imputation to agree with the observed data, in which case $\mathcal{C} = \{\mathbf{w}: w_{t} = x_{t}, t \notin \mathcal{M}\}$. 
It may also reflect the side information available.
For example, if $x_{t}$ are monthly sales data, the knowledge of quarterly sales define a set of linear constraints for the missing values. 
Another example is when $\{\mathbf{x}_{t}\}$ is a $d$-dimensional compositional time series such as the market share of several companies in an industry. 
It follows that $\mathbf{x}_{t}^{\top}\mathbf{1} = 1$ for all $t$. 
Then, by defining $\mathcal{C} = \{\mathbf{W} \in \mathbb{R}^{n \times d}: \mathbf{W}\mathbf{1}_{d} = \mathbf{1}_{n}\}$, where $\mathbf{1}_{n} = (1,\ldots,1)^{\top} \in \mathbb{R}^{n}$, the set of admissible imputations are restricted to be compositional as well. 

One may also consider the following regularized version of \eqref{sec2_what}.
\begin{align} \label{sec2_whatreg}
    \hat{\mathbf{w}} := (\hat{w}_{0}, \ldots \hat{w}_{n-1})^{\top} = \arg\min_{\mathbf{w} \in \mathcal{C}} \left\{ \mathsf{OTC}(\hat{\mu}_{\mathrm{pre}}(\mathbf{w}), \hat{\mu}_{\mathrm{post}}(\mathbf{w})) + \frac{\lambda}{2} \Vert \mathbf{w} \Vert^{2} \right\},    
\end{align}
where the regularization parameter $\lambda > 0$ is small. 
As we will see in the next subsection, the $\ell_{2}$-regularization term enhances the numerical stability of the alternating minimization algorithm.
Since $L$ is often chosen to induce the Wasserstein distances between the marginal distributions of a time series, the solution $\hat{\mathbf{w}}$ to \eqref{sec2_what} or \eqref{sec2_whatreg} is called the temporal Wasserstein imputation (TWI) throughout this paper.

In this formulation, it is clear that the temporal Wasserstein imputation, which seeks to equate the marginal distributions on both sides of $n_{1}$, is nonparametric.
With a mildly large $p$, the $p$-dimensional marginal distributions can sometimes fully characterize the nonlinear dynamics of a Markov process (such as a finite-order threshold AR process), which would otherwise need many parameters using ARMA approximation.

\subsection{Alternating minimization algorithm}

The optimization landscapes of \eqref{sec2_what} and \eqref{sec2_whatreg} are complicated because the optimal transport cost is itself the minimum value of a linear program \eqref{sec2_OTC}, and there are potentially many variables to optimize. 
To tackle this issue, we propose to solve the nested minimization problem using alternating minimization.
Since \eqref{sec2_what} is a special case of \eqref{sec2_whatreg} with $\lambda = 0$, we focus our discussion on \eqref{sec2_whatreg}.

To begin, observe that with an initialized $\hat{\mathbf{w}}^{(0)}$, problem \eqref{sec2_OTC} is exactly the discrete optimal transport problem. 
Thus, at iteration $k=1,2,\ldots$, treating $\mathbf{w} = \hat{\mathbf{w}}^{(k-1)}$ as fixed, we may compute the optimizer $\hat{\mathbf{\Pi}}^{(k)}$ of \eqref{sec2_OTC} efficiently by existing algorithms \citep{Cuturi2013, peyre2019}. 
This corresponds to the inner optimization problem of \eqref{sec2_whatreg}.
Next, we optimize \eqref{sec2_whatreg} over $\mathbf{w}$, taking $\mathbf{\Pi} = \hat{\mathbf{\Pi}}^{(k)}$ as fixed. 
This procedure is repeated until convergence (discussed in the next section).
Let 
\begin{align} \label{Sec3_2-F}
    F(\mathbf{w}, \mathbf{\Pi}) = \sum_{i=p-1}^{n_{1}} \sum_{j=n_{1}+1}^{n-1} \pi_{ij} L(\mathbf{v}_{i}(\mathbf{w}), \mathbf{v}_{j}(\mathbf{w})) + \frac{\lambda}{2} \Vert \mathbf{w} \Vert^{2},
\end{align}
and the procedure described above can be summarized in Algorithm \ref{alg:WI}. 

\spacingset{1}
\begin{algorithm}[t]
\DontPrintSemicolon
    \KwInput{Cut-off $n_{1}$, number of lags $p$}
    \KwInit{$\mathbf{w}^{(0)} \in \mathcal{C}$}
    \For{$k = 1,2, \ldots$}{
        Step (a): Solve for $\hat{\mathbf{\Pi}}^{(k)}$ in problem \eqref{sec2_OTC}, treating $\mathbf{w} = \hat{\mathbf{w}}^{(k-1)}$ as fixed.  \\
        Step (b): Solve for $\hat{\mathbf{w}}^{(k)} \in \arg\min_{\mathbf{w} \in \mathcal{C}} F(\mathbf{w}, \hat{\mathbf{\Pi}}^{(k)})$. \\
        \textbf{Break} if convergence criterion is met and set $\hat{\mathbf{w}} = \hat{\mathbf{w}}^{(k)}$
    }    
    \KwOutput{Imputed time series $\hat{\mathbf{w}}$}
\caption{Temporal Wasserstein imputation with alternating minimization}
\label{alg:WI}
\end{algorithm}
\spacingset{2}

One advantage of the alternating minimization scheme is that the subproblems (steps (a) and (b) in Algorithm \ref{alg:WI}) are easy to optimize when $L$ satisfies some regularity conditions.
Suppose $L(\mathbf{v}_{i}, \mathbf{v}_{j}) = \Vert \mathbf{v}_{i} - \mathbf{v}_{j} \Vert^{2}$.
Then it is not difficult to show that $F(\mathbf{w}, \mathbf{\Pi}) = \mathbf{w}^{\top} \mathbf{H}(\mathbf{\Pi}) \mathbf{w}$, where $\mathbf{H}(\mathbf{\Pi}) = \sum_{h=0}^{p-1}\mathbf{A}_{h}(\mathbf{\Pi}) + \frac{\lambda}{2} \mathbf{I}_{n}$ and
\begin{align*}
    \mathbf{A}_{h}(\mathbf{\Pi}) &= 
    \begin{pmatrix}
        \mathbf{0}_{p-1-h} & & & \\
        & \frac{1}{n_{1}-p+2}\mathbf{I}_{n_{1}-p+2} & -\mathbf{\Pi} & \\
        & -\mathbf{\Pi}^{\top} & \frac{1}{n - n_{1} - 1}\mathbf{I}_{n-n_{1}-1} & \\
        & & & \mathbf{0}_{h}
    \end{pmatrix} \in \mathbb{R}^{n \times n}, \quad h = 0, 1, \ldots, p-1,
\end{align*}
in which all unspecified entries in $\mathbf{A}_{h}(\mathbf{\Pi})$ are zero. 
Hence $F$ is quadratic in $\mathbf{w}$.
In fact, it can be shown that $\mathbf{H}(\mathbf{\Pi})$ is nonnegative definite for any $\mathbf{\Pi}$ satisfying the coupling constraints, and $F$ is strongly convex in $\mathbf{w}$ if $\lambda > 0$ (see also Proposition \ref{prop::biconvex} in Section \ref{Sec::theory}). 
Thus many efficient solvers are applicable in this step \citep{boyd2004convex}. 
If, in addition, the constraint set $\mathcal{C}$ is defined through a system of linear equations, i.e.,
    $\mathcal{C} = \{\mathbf{w} \in \mathbb{R}^{n}: \mathbf{Kw} = \mathbf{b}\}$
where $\mathbf{K}$ and $\mathbf{b}$ are known, then step (b) in Algorithm \ref{alg:WI} has the closed form
\begin{align} \label{Sec3_closedform}
    \hat{\mathbf{w}}^{(k)} = \mathbf{H}^{-1}\mathbf{K}^{\top}(\mathbf{K}\mathbf{H}^{-1}\mathbf{K}^{\top})^{-1}\mathbf{b},
\end{align}
where $\mathbf{H} = \mathbf{H}(\hat{\mathbf{\Pi}}^{(k)})$.
Suppose further that $\mathcal{C} = \{\mathbf{w}: w_{t} = x_{t}, t \notin \mathcal{M}\}$, corresponding to the case of no side information. 
The solution $(\hat{w}_{0}, \ldots, \hat{w}_{n-1})$ to the subproblem in step (b) must satisfy the following system of linear equations. 
\begin{equation} \label{Sec3_2-1}
\left\{
\begin{aligned}
    \hat{w}_{s} =& \frac{\sum_{i=s\vee(p-1)}^{s+p-1} \sum_{n_{1}+1}^{n-1} \hat{\pi}_{ij}\hat{w}_{j-i+s}}{\frac{p\wedge(s+1)}{n_{1}-p+2} + \frac{\lambda}{2}}, \quad \mbox{if } 0 \leq s \leq n_{1}-p+1, s \in \mathcal{M} \\
    \hat{w}_{s} =& \frac{\sum_{i=s}^{n_{1}} \sum_{j=n_{1}+1}^{n-1} \hat{\pi}_{ij}\hat{w}_{j-i+s}  + \sum_{i=p-1}^{n_{1}} \sum_{j=n_{1}+1}^{s+p-1} \hat{\pi}_{ij} \hat{w}_{i-j+s}}{\frac{n_{1}-s+1}{n_{1}-p+2} + \frac{s+p-n_{1}+1}{n-n_{1}-1} + \frac{\lambda}{2}}, \\
    & \hspace{5cm} \mbox{if } n_{1}-p+2 \leq s \leq n_{1}, s \in \mathcal{M} \\
    \hat{w}_{s} =& \frac{\sum_{i=p-1}^{n_{1}} \sum_{j=s}^{(s+p-1)\wedge(n-1)} \hat{\pi}_{ij}\hat{w}_{i-j+s}}{\frac{p \wedge (n-s)}{n - n_{1} - 1} + \frac{\lambda}{2}}, \quad \mbox{if } n_{1}+1 \leq s \leq n-1, s \in \mathcal{M} \\
    \hat{w}_{s} =& x_{s},\quad \mbox{if } s \notin \mathcal{M}
\end{aligned}
\right.
\end{equation}
where $\hat{\mathbf{\Pi}} = (\hat{\pi}_{ij})$ is treated as fixed, and for $a, b \in \mathbb{R}$, $a \wedge b = \min\{a, b\}$, $a \vee b = \max\{a, b\}$. We make a few remarks regarding \eqref{Sec3_2-1}.
\begin{remark} \label{remark-1}
    \normalfont  Let us contrast TWI with the optimal interpolator. Equation \eqref{Sec3_2-1} shows, in general, the temporal Wasserstein imputation is a linear function of \textit{many} observed data $\{x_{t}: t \notin \mathcal{M}\}$, even when $p$ is small.
    In contrast, the optimal interpolator is often a function of nearby observed data.
    For instance, suppose the underlying time series is a stationary AR(1) process: $x_{t} = \phi x_{t-1} + \epsilon_{t}$, where $\{\epsilon_{t}\}$ is a mean-zero, unit-variance white noise.
    If $x_{n_{1}} = \bigstar$ is the only missing value in the time series, then the optimal interpolator \citep{pena1991} of $x_{n_{1}}$ is $w_{n_{1}, \mathrm{opt}} = \phi(x_{n_{1}-1} + x_{n_{1}+1})/(1 + \phi^{2})$.
    In contrast, \eqref{Sec3_2-1} shows that the temporal Wasserstein imputation (with $p=2$ and $\lambda=0$) is given by 
    \begin{align*}
        \hat{w}_{n_{1}} = \frac{1}{n_{1}^{-1} + (n - n_{1} - 1)^{-1}}\left( \sum_{t=0}^{n_{1}-1} \hat{\pi}_{t+1,n_{1}+1}x_{t} + \sum_{t=n_{1}+1}^{n-1}\hat{\pi}_{n_{1},t}x_{t} \right),
    \end{align*}
    which is a convex combination of all observed values.
    This is expected because, as a nonparametric method, TWI uses all data to learn the underlying dynamics.
    Note that convex combination appears naturally in imputation problems. For instance, in synthetic control, a convex combination of untreated units is used to estimate the counterfactual (see \citealp{hur2024, Abadie2021}).
\end{remark}

\begin{remark}
    \normalfont (EM interpretation) Let $(\hat{\mathbf{w}}, \hat{\mathbf{\Pi}})$ be a fixed point of the alternating minimization algorithm (Algorithm \ref{alg:WI}). That is, $\hat{\mathbf{w}} \in \arg\min_{\mathbf{w} \in \mathcal{C}} F(\mathbf{w}, \hat{\mathbf{\Pi}})$, and $\hat{\mathbf{\Pi}} \in \arg\min_{\mathbf{\Pi}} F(\hat{\mathbf{w}}, \mathbf{\Pi})$.
    In addition, suppose $\lambda = 0$ for a moment. 
    Let $\mathbf{u} = (u_{0}, u_{1}, \ldots, u_{p-1})$ be a random vector whose distribution is $\hat{\mu}_{\mathrm{pre}}$ and $\mathbf{v} = (v_{0}, \ldots, v_{p-1})$ a random vector whose distribution is $\hat{\mu}_{\mathrm{post}}$. 
    Furthermore, let the joint distribution of $(\mathbf{u}, \mathbf{v})$ be given by the optimal transport coupling between $\hat{\mu}_{\mathrm{pre}}$ and $\hat{\mu}_{\mathrm{post}}$, identified as $\hat{\mathbf{\Pi}}$ (viewed as a joint distribution between the two discrete measures).
    Consider an index $s \in \mathcal{M}$ with $p-1 \leq s \leq n_{1}-p+1$.
    By \eqref{Sec3_2-1}, it is not difficult to see that
    \begin{align}
        \hat{w}_{s} =& \frac{1}{p} \sum_{i=s}^{s+p-1} \sum_{j=n_{1}+1}^{n-1} \frac{\hat{\pi}_{ij}}{\sum_{j=n_{1}+1}^{n-1} \hat{\pi}_{ij}} \hat{w}_{j-i+s} \notag \\
        =& \frac{1}{p} \sum_{h=0}^{p-1} \mathbb{E}_{\hat{\mathbf{\Pi}}}(v_{h}|\mathbf{u} = (\hat{w}_{s+h}, \ldots, \hat{w}_{s}, 
        \ldots, \hat{w}_{s+h-p+1})). \label{Sec3_EM-2}
    \end{align}
    Equation \eqref{Sec3_EM-2} shows that the temporal Wasserstein imputation $\hat{w}_{s}$ is essentially an average of conditional expectations.
    Thus, Algorithm \ref{alg:WI} is closely related to the EM algorithm: To minimize the Wasserstein distance $\mathbb{E}_{\mathbf{\Pi}}\Vert \mathbf{u} - \mathbf{v}\Vert^{2}$, step (b) computes the conditional expectations under the current optimal transport coupling $\hat{\mathbf{\Pi}}$ while step (a) minimizes over the coupling space.  
    Similar formulas to \eqref{Sec3_EM-2} can be derived for other time indices $s$ with appropriate modifications.
\end{remark}


Because of the optimization formulation, it is straightforward to extend TWI to multivariate time series. 
Suppose the multivariate time series of interest is $\mathbf{x}_{t} = (x_{t,1}, \ldots, x_{t,d})^{\top}$.
To compute the imputation $\mathbf{W} = (\mathbf{w}_{0}, \mathbf{w}_{1}, \ldots, \mathbf{w}_{n-1})$ where $\mathbf{w}_{t} = (w_{t,1}, \ldots, w_{t,d})^{\top}$, we may replace $\mathbf{v}_{i}(\mathbf{w})$ in \eqref{sec2_OTC} by $\mathbf{v}_{i}(\mathbf{W}) := (w_{i,1}, \ldots, w_{i-p+1,1}, w_{i,2}, \ldots, w_{i-p+1,2}, \ldots, w_{i-p+1,d})^{\top}$.
Then Algorithm \ref{alg:WI} can still be applied. 
Since the multivariate information is summarized by the distance matrix $\mathbf{L} = (l_{ij})_{0 \leq i \leq n_{1}, n_{1}+1 \leq j \leq n-1}$, where $l_{ij} = L(\mathbf{v}_{i}(\mathbf{W}), \mathbf{v}_{j}(\mathbf{W}))$, regardless of the dimension $d$, step (a) in Algorithm \ref{alg:WI} remains unchanged. 
In the notable special case where $L(\mathbf{u}, \mathbf{v}) = \Vert \mathbf{u} - \mathbf{v} \Vert^{2}$, \eqref{Sec3_2-F} becomes
\begin{align*}
    F(\mathbf{W}, \mathbf{\Pi}) = \sum_{l=1}^{d} \sum_{i = p-1}^{n_{1}} \sum_{j=n_{1}+1}^{n-1} \pi_{ij} \left[ \sum_{h=0}^{p-1} (w_{i-h,l} - w_{j-h,l})^{2} \right] + \frac{\lambda}{2} \Vert \mathbf{W} \Vert_{F}^{2},
\end{align*}
where $\Vert \cdot \Vert_{F}$ is the Frobenius norm.
Straightforward calculations show that fixing $\mathbf{\Pi} = \hat{\mathbf{\Pi}}$, if $\hat{\mathbf{W}}$ is a solution to step (b) of Algorithm \ref{alg:WI}, then for each $l \in \{1,2,\ldots, d\}$, $(\hat{w}_{0,l}, \hat{w}_{1,l}, \ldots, \hat{w}_{n-1,l})$ also satisfies the system of linear equations \eqref{Sec3_2-1} in the case of no side information.
Therefore, in step (b) of Algorithm \ref{alg:WI}, each component in the multivariate time series can be imputed in parallel, lowering the computational overhead. 
Thus, the algorithm is fast to implement in the multivariate setting.

Before closing this subsection, we discuss the imputation of nonstationary time series.
For an $I(1)$ process $\{x_{t}\}$, $y_{t} = x_{t} - x_{t-1}$ is stationary and one can apply imputation methods to the first-differenced series $\{y_{t}\}$.
More generally, one can apply TWI to the series $y_{t} = \phi(B) x_{t}$, where $\phi(z) = (1-z)^{d}$ if $x_{t}$ is an $I(d)$ process, and $B$ is the back-shift operator.
Then, an imputation for $x_{t}$ can be recovered by inverting $\phi(B)$. 
However, in this case, the number of missing data in $\{y_{t}\}$ multiplies, and at least $d+1$ contiguous observed values in $\{x_{t}\}$ is needed to obtain an observed value for $\{y_{t}\}$.
To mitigate the information loss, we can encode the information from the raw data via an admissible set $\mathcal{C}$.
For example, consider the $I(1)$ case: If $x_{t}$ is observed, then since $y_{1} + \ldots + y_{t} = x_{t} - x_{0}$, we may restrict the imputation $\{w_{t}\}$ for $\{y_{t}\}$ to satisfy $w_{1}+\ldots+w_{t} = x_{t}- x_{0}$ (assuming $x_{0}$ is observed), which constitutes a set of linear constraints $\mathcal{C}$.
Finally, although we have motivated TWI using the stationarity of the underlying process, it fares quite well with certain nonstationary processes, in addition to unit-root series, such as cyclic time series (see Section \ref{Sec::Simulation} for a simulation study).

\subsection{Related works} \label{Subsec::relatedworks}

We shall compare TWI with some existing nonparametric imputation approaches.
\citet{Carrizosa2013} proposed a moment-matching method, which tries to minimize the discrepancies between the moments (e.g.~autocorrelations) of the imputed series and some pre-specified target values. 
However, their non-convex algorithm is computationally intensive, and good target values for the moments may not be readily available in practice.
In contrast, TWI is completely data-driven, and enjoys the computational efficiency of the alternating minimization algorithm and computational optimal transport.

Singular spectrum analysis (SSA) methods \citep{golyandina2018, golyandina2001, Kondrashov2006} are nonparametric techniques that perform spectral analysis on the trajectory matrix of the time series $\mathbf{x} = (x_{0}, x_{1}, \ldots, x_{n-1})$. Specifically, the trajectory matrix is defined as
$\mathbf{T} = (\mathbf{y}_{L-1}, \mathbf{y}_{L}, \ldots, \mathbf{y}_{n-1})$, where $\mathbf{y}_{j} = (x_{j-L+1}, x_{j-L+2}, \ldots, x_{j})^{\top}$ for $j=L-1, \ldots, n-1$, and 
$L$ is called the window length. SSA methods exploit the possible low-rank structure of $\mathbf{T}$.
When $\mathbf{T}$ contains missing values, columns with missing entries are typically removed (shaped SSA, \citealp{golyandina2018}), followed by direct or iterative imputation \citep{Kondrashov2006}.
For some missing patterns this could result in too few complete columns, forcing $L$ to be small (violating a common guideline of $L \approx n / 2$) and compromising the effectiveness of SSA.
In contrast, TWI accommodates any missing pattern without discarding substantial data.
While SSA is typically effective for imputing structural time series, TWI prevails in learning stochastic time series (see Section \ref{Sec::Simulation} for our simulation study).
However, due to the convex nature of TWI (see Remark \ref{remark-1}), it cannot handle series with explosive deterministic trends.
Hence, it is advisable to remove the trends before applying TWI, or use SSA-based methods when the deterministic patterns are of primary interest.

A concurrent and independent work of \citet{wang2025} proposed the Proximal Spectrum Wasserstein for Imputation (PSW-I), which, following the framework of \citet{Muzellec2020}, imputes time series using optimal transport techniques.
While both PSW-I and TWI focus on distribution matching, their implementations differ significantly.
PSW-I uses a gradient descent algorithm that relies on automatic differentiation and some ad hoc gradient-stabilizing techniques, due to the difficulty in computing the gradient of the (regularized) Wasserstein distance with respect to the imputation. 
TWI, by comparison, uses a straightforward algorithm with minimal tuning. 
Although \citet{wang2025} includes a convergence result, its guarantees hinge on restrictive assumptions. In contrast, TWI’s convergence and statistical properties are established under more standard conditions, as detailed in the next section.


\section{Theoretical properties} \label{Sec::theory}
This section investigates the optimization and statistical properties of TWI. 
We first show Algorithm \ref{alg:WI} converges to critical points.
We then establish the asymptotic consistency of TWI when the missing data concentrate on one side of $n_{1}$.
For general non-systematic missing patterns, a case study of a two-state Markov chain example shows that TWI can identify the underlying marginal distributions under some natural assumptions. 

\subsection{Optimization and statistical guarantees} \label{Subsec::optim_theory}

We begin by studying the optimization properties of Algorithm \ref{alg:WI}. All proofs are relegated to Section S1 of the supplementary material.
The following result shows Algorithm \ref{alg:WI} is guaranteed to find a local minimum. 
\begin{prop} \label{prop::biconvex}
Suppose $\mathcal{C}$ is a convex set and $F(.,.)$ is defined in (\ref{Sec3_2-F}). Then,
\begin{itemize}
    \item[(a)] (monotonicity) $F(\hat{\mathbf{w}}^{(t)}, \hat{\mathbf{\Pi}}^{(t+1)}) \leq F(\hat{\mathbf{w}}^{(t)}, \hat{\mathbf{\Pi}}^{(t)}) \leq F(\hat{\mathbf{w}}^{(t-1)}, \hat{\mathbf{\Pi}}^{(t)})$, $t=1,2,\ldots$;
    \item[(b)] every limit point $(\hat{\mathbf{w}}, \hat{\mathbf{\Pi}})$ of $\{(\hat{\mathbf{w}}^{(t+1)}, \hat{\mathbf{\Pi}}^{(t)}):t=1,2,\ldots\}$ is a critical point of $F$. i.e.,
    \begin{align*}
        \frac{\partial F(\hat{\mathbf{w}}, \hat{\mathbf{\Pi}})}{ \partial\mathbf{w}} \cdot (\mathbf{w} - \hat{\mathbf{w}}) \geq 0, \quad
        \frac{\partial F(\hat{\mathbf{w}}, \hat{\mathbf{\Pi}})}{ \partial \mathrm{vec}(\mathbf{\Pi})} \cdot \mathrm{vec}(\mathbf{\Pi} - \hat{\mathbf{\Pi}}) \geq 0
    \end{align*}
    for any $\mathbf{w} \in \mathcal{C}$ and $\mathbf{\Pi}$ satisfying the coupling constraints. 
\item[(c)] If, in addition, the cost function $L$ takes the form $L(\mathbf{u}, \mathbf{v}) = \sum_{k=1}^{p} \ell(u_{k} - v_{k})$, where $\mathbf{u} = (u_{1}, \ldots, u_{p})^{\top}$, $\mathbf{v} = (v_{1}, \ldots, v_{p})^{\top}$, and $\ell$ is twice differentiable and convex, then 
    $F(\mathbf{w}, \mathbf{\Pi})$ is biconvex in $\mathbf{w}$ and $\mathbf{\Pi}$.
\end{itemize}
\end{prop}

Due to the biconvexity of $F$, one can employ off-the-shelf solvers for the subproblems in Algorithm \ref{alg:WI}.
Moreover, it follows directly from the proof that if $\lambda > 0$, $F(\mathbf{w}, \mathbf{\Pi})$ is $\lambda$-strongly convex in $\mathbf{w}$. Hence, step (b) in Algorithm \ref{alg:WI} admits a unique solution in each iteration and can be solved efficiently. 
It is straightforward to generalize Proposition \ref{prop::biconvex} for multivariate time series, as discussed in Section \ref{Sec::Method}.

Unlike typical non-convex optimization problems, we are only interested in the subset of the global minimizers of $F$ which have the correct marginal distributions 
(see Section \ref{Subsec::markov}). 
The problem simplifies when the missing data happen to lie on the same side of $t=n_{1}$, in which case the global minimizer of $F$ achieves asymptotic consistency. 

\begin{prop} \label{prop::oneside}
    Assume (A1) $\{\tilde{x}_{t}\}$ is stationary and ergodic, and (A2) $\mathbb{E}|\tilde{x}_{t}|^{k}<\infty$ for some $k \geq 1$. 
    Let $\mathcal{M} = \mathcal{M}_{n}$ be the set of indices corresponding to the missing data. Suppose for any $\mathbf{w} \in \mathcal{C}$, $w_{t} = \tilde{x}_{t}$ if $t \notin \mathcal{M}_{n}$.  
    Suppose further $n$ and $n_{1}$ satisfy $n_{1} \rightarrow \infty$ and $n - n_{1} \rightarrow \infty$.
    If either $\mathcal{M}_{n} \subseteq \{0,1,\ldots, n_{1}\}$ or $\mathcal{M}_{n} \subseteq \{n_{1}+1, \ldots, n-1\}$ (but not both),
    then, the global minimizer $(\hat{\mathbf{w}}, \hat{\mathbf{\Pi}})$ of $F$ with $\lambda = 0$ and $L(\mathbf{u}, \mathbf{v}) = \Vert \mathbf{u} - \mathbf{v} \Vert^{k}$, subject to $\mathbf{w} \in \mathcal{C}$ and the coupling conditions, satisfies
    \begin{align*}
        \mathcal{W}_{k}\left( \frac{1}{n_{1}-p+2}\sum_{t=p-1}^{n_{1}} \delta_{\mathbf{v}_{t}(\hat{\mathbf{w}})}, \mu_{p} \right) + \mathcal{W}_{k}\left( \frac{1}{n-n_{1}-1}\sum_{t=n_{1}+1}^{n-1} \delta_{\mathbf{v}_{t}(\hat{\mathbf{w}})}, \mu_{p} \right) \rightarrow 0 \quad \mbox{almost surely}.
    \end{align*}
\end{prop}

Assumption (A1) is commonly used in nonlinear time series analysis \citep{tsay2018}, and the moment condition (A2) is  mild. 
Note that the missing ratio can be asymptotically non-negligible: $\liminf_{n \rightarrow \infty}\sharp(\mathcal{M}_{n})/n > 0$.
In fact, this case is of primary interest because distributional inconsistency would be precluded if $\sharp(\mathcal{M}_{n})/n \rightarrow 0$.
Proposition \ref{prop::oneside} implies the consistency of the empirical marginal distribution of the imputed time series, i.e.,~
$\hat{\mu}_{\mathrm{pre}}(\hat{\mathbf{w}}) \Rightarrow \mu_{p}$ and $\hat{\mu}_{\mathrm{post}}(\hat{\mathbf{w}}) \Rightarrow \mu_{p}$.
Moreover, convergence in Wasserstein distance ensures for any continuous $\varphi$ with $|\varphi(x)| \leq C(1+\Vert x \Vert^{k})$ for all $x \in \mathbb{R}^{p}$ and some $C < \infty$, we have $\int \varphi d\hat{\mu}_{n} \rightarrow \int \varphi d\mu$,
where $\hat{\mu}_{n}$ is either $\hat{\mu}_{\mathrm{pre}}(\hat{\mathbf{w}})$ or $\hat{\mu}_{\mathrm{post}}(\hat{\mathbf{w}})$. 
For such functionals, TWI provides consistent downstream statistics, as summarized in the following corollary.

\begin{corollary}
    Assume that the conditions of Proposition \ref{prop::oneside} hold. Then for continuous $\varphi:\mathbb{R}^{p} \rightarrow \mathbb{R}$ with $|\varphi(x)| \leq C(1+\Vert x \Vert^{k})$ for all $x \in \mathbb{R}^{p}$ and some $C < \infty$, we have 
    $\frac{1}{n-p+1} \sum_{t = p-1}^{n-1} \varphi(\hat{w}_{t}, \ldots, \hat{w}_{t-p+1}) \rightarrow \mathbb{E}\varphi(\tilde{x}_{t}, \ldots, \tilde{x}_{t-p+1})$ a.s., where $\hat{\mathbf{w}}=(\hat{w}_{0}, \ldots, \hat{w}_{n-1})$ is defined in Proposition \ref{prop::oneside}.
\end{corollary}


\subsection{General missing pattern: A two-state Markov chain example} \label{Subsec::markov}
In this subsection, we study the identification of the underlying marginal distributions via a two-state Markov chain example under the general missing pattern where missing data are not confined to one side of time index $n_{1}$.
We connect the non-identification problem to the lack of uniqueness of the solutions to a system of linear equations.
We also show that, under two natural conditions, the correct marginal distributions can still be pinned down by TWI.
To free our discussion from finite-sample complications, we will focus on the distribution level in this subsection, or intuitively speaking, assume ``infinite sample.''

Suppose the time series $\{\tilde{x}_{t}\}_{-\infty < t < \infty}$ under study is a two-state Markov chain.
Specifically, $\tilde{x}_{t} \in \{0, 1\}$ and the transition probabilities are given by $\mathbb{P} (\tilde{x}_{t} = 1 | \tilde{x}_{t-1} = 0) = p$, and $\mathbb{P} (\tilde{x}_{t} = 0 | \tilde{x}_{t-1} = 1) = q$,
where $p, q \in (0, 1)$.
The associated stationary distribution is denoted by $\lambda_{1} := \mathbb{P}(\tilde{x}_{t} = 0)$ and $\lambda_{2} := \mathbb{P}(\tilde{x}_{t} = 1)$.
Consequently, the 2-dimensional marginal distribution of $\{\tilde{x}_{t}\}$ is 
\begin{equation} \label{Sec4_2-1}
\begin{array}{llll}
    \mathbb{P} ((\tilde{x}_{t}, \tilde{x}_{t-1}) = (1,1)) = \lambda_{2}(1 - q),  &
    \mathbb{P} ((\tilde{x}_{t}, \tilde{x}_{t-1}) = (1,0)) = \lambda_{1}p, \\
    \mathbb{P} ((\tilde{x}_{t}, \tilde{x}_{t-1}) = (0,1)) = \lambda_{2}q,        &
    \mathbb{P} ((\tilde{x}_{t}, \tilde{x}_{t-1}) = (0,0)) = \lambda_{1}(1-p).
\end{array}\end{equation}

We consider the following missing pattern.
Without loss of generality, set $n_{1} = 0$. 
Suppose that before time $n_{1}$, there is a missing value every $k_{1}$ observations. 
Likewise, there is a missing value every $k_{2}$ observations after time $n_{1}$.
Specifically, $x_{t} = \bigstar$ if $-t$ is a positive integer multiple of $k_{1}$, or if $t$ is a positive integer multiple of $k_{2}$, where $k_{1}$ and $k_{2}$ are some integers larger than 2 and $k_1\neq k_2$ (see condition (C1) below).
Next, we parametrize an imputation $\{\hat{w}_{t}\}_{-\infty < t < \infty}$ as follows. 
Specifically, let
\begin{equation} \label{Sec4_2-2}
\begin{aligned}
    &\mathbb{P}(\hat{w}_{t}=1|(x_{t}, x_{t-1}) = (\bigstar, 1)) = a_{1}, 
    &\mathbb{P}(\hat{w}_{t}=1|(x_{t}, x_{t-1}) = (\bigstar, 0)) = b_{1},  
    &\mbox{ if } t < n_{1},
\end{aligned}
\end{equation}
\begin{equation} \label{Sec4_2-3}
\begin{aligned}
    &\mathbb{P}(\hat{w}_{t}=1|(x_{t}, x_{t-1}) = (\bigstar, 1)) = a_{2}, 
    &\mathbb{P}(\hat{w}_{t}=1|(x_{t}, x_{t-1}) = (\bigstar, 0)) = b_{2},  
    & \mbox{ if } t \geq n_{1},
\end{aligned}
\end{equation}
for some $a_{1}, a_{2}, b_{1}, b_{2} \in [0,1]$.
Essentially, we parametrize all possible imputations by $a_{1}, a_{2}, b_{1}, b_{2}$, each corresponding to different 2-dimensional marginal distributions.
The correct imputation should calibrate $a_{i}$'s and $b_{i}$'s such that the resulting marginal distributions equal those of the underlying series, given in \eqref{Sec4_2-1}. 
It is not difficult to see that the correct imputation must choose
\begin{align} \label{Sec4_2-3.5}
    a_{1}^{*} = a_{2}^{*} = 1 - q, \quad
    b_{1}^{*} = b_{2}^{*} = p.    
\end{align}
Note that this calibration does not depend on the missing pattern $k_{1}, k_{2}$. 

Assuming the optimization is carried out perfectly\footnote{\spacingset{1}\footnotesize This is where we ignore finite sample complications and assume a zero Wasserstein loss is achieved. One may use, for example, the hamming distance 
as the ground metric for the optimal transport cost.}, TWI calibrates $a_{i}, b_{i}, i = 1, 2$, by equating the 2-dimensional marginal distributions of $\{\hat{w}_{t}\}_{-\infty<t<0}$ and $\{\hat{w}_{t}\}_{0\geq t<\infty}$.
Hence, $\{a_{1}, a_{2}, b_{1}, b_{2}\}$ satisfies
\begin{equation} \label{Sec4_2-4}
\left.
\begin{aligned}
    \frac{1}{k_{1}}a_{1} - \frac{1}{k_{2}}a_{2} =& (1 - q)(\frac{1}{k_{1}} - \frac{1}{k_{2}}), \\
    \frac{1}{k_{1}}b_{1} - \frac{1}{k_{2}}b_{2} =& p(\frac{1}{k_{1}} - \frac{1}{k_{2}}).
\end{aligned}
\right\}
\end{equation}
Derivation of \eqref{Sec4_2-4} is relegated to the Section S1 of the supplementary material.
Now we introduce two conditions.
\begin{description}
    \item[(C1)] Non-systematic missing pattern: $k_{1} \neq k_{2}$. 
    \item[(C2)] Interpolator stability: $a_{1} = a_{2}$, $b_{1} = b_{2}$.
\end{description}
Clearly, identification is not possible when $k_{1} = k_{2}$, in which case any imputation with $a_{1} = a_{2}$ and $b_{1} = b_{2}$ would satisfy \eqref{Sec4_2-4}, so (C1) is necessary for identification. Intuitively, (C1) precludes the special case of down-sampling, which is notoriously hard to be treated as a missing data problem.
However, even with (C1) the solutions to \eqref{Sec4_2-4} are not unique.
This lack of uniqueness is at the crux of the identification problem for time series in the distribution matching framework \citep{Muzellec2020, wang2025}.
There are many imputations whose marginal distributions before and after $n_{1}$ are the same, but they may not equal to the one given in \eqref{Sec4_2-1}.
Inspired by \eqref{Sec4_2-3.5}, (C2) imposes an additional requirement for the ideal imputation, which can be verified without knowledge about the missing data.
It requires that, given the same neighboring observations, one should impute the missing entries in the same way (at least probabilistically) before and after $n_{1}$.
Together, (C1) and (C2) imply a unique solution to the system \eqref{Sec4_2-4}, which guarantees identification. 

\begin{prop} \label{prop:Markov}
    Under (C1) and (C2), the only solution to \eqref{Sec4_2-4} is $a_{1} = a_{2} = a_{1}^{*}$, $b_{1} = b_{2} = b_{1}^{*}$.
    That is, the \textbf{only} 2-dimensional marginal distribution that achieves zero Wasserstein loss and satisfies (C2) is that of the underlying distribution.
\end{prop}

As shown in Proposition \ref{prop:Markov}, it is not necessary to use an $n_{1}$ such that all missing data lie on the same side.
In practice, appropriate tuning of $n_{1}$ and $p$ is needed. 
We can select $n_{1}$ from a pre-specified interval $[t_{1}, t_{2}]$ which yields the smallest Wasserstein loss, where $t_{1}$ and $n-t_{2}$ are sufficiently large.
Alternatively, a simple procedure to promote interpolator stability is to run TWI with different splitting cut-offs, using previously obtained imputations as initializations. This procedure, which we call $k$-TWI, is summarized in Algorithm \ref{alg:kWI}. 
Finally, similar to other nonparametric methods, $p$ can be chosen to diverge with the sample size so long as $p/n \rightarrow 0$.


\spacingset{1}
\begin{algorithm}[t]
\DontPrintSemicolon
    \KwInput{Cut-off points $n_{1}, n_{2}, \ldots, n_{k}$, number of lags $p$}
    \KwInit{$\hat{\mathbf{w}}^{(0)} \in \mathcal{C}$}
    \For{$h = 1,2, \ldots, k$}{
        Run Algorithm \ref{alg:WI} with initialization $\hat{\mathbf{w}}^{(0)}$ and cut-off $n_{h}$ to obtain $\hat{\mathbf{w}}$ \\
        Set $\hat{\mathbf{w}}^{(0)} = \hat{\mathbf{w}}$
    }    
    \KwOutput{Imputed time series $\hat{\mathbf{w}}$}
\caption{$k$-TWI}
\label{alg:kWI}
\end{algorithm}
\spacingset{2}

\section{Simulation Studies} \label{Sec::Simulation}
In this section, we apply TWI to data generated from various time series models, including multivariate, nonlinear, and nonstationary ones, to examine its finite sample performance.
For comparison, we employ Kalman smoothing, the iterative gap filling by SSA \citep{Kondrashov2006}, the scalar filtering algorithm \citep{pena2021}, and linear interpolations. 
The following two missing patterns are considered: 
\begin{itemize}
    \item[] \textbf{Pattern I}: 300 observations are omitted at random;
    \item[] \textbf{Pattern II}: For every 20 observations, 6 consecutive observations are omitted.
\end{itemize}
Pattern I simulates data missing at random while Pattern II generates patches of missing data within the same positions of every 20 observations.
The total length of the time series is set to $n = 1000$. Thus, for both missing patterns, 30\% of the data are missing. 
For each model and missing pattern considered, 1000 Monte Carlo simulations are carried out.

The proximal gradient descent \citep{parikh2014proximal} is used to efficiently solve step (b) in Algorithm \ref{alg:WI} when no side information is available because the projection step is straightforward. 
When $\mathcal{C}$ is defined by a system of linear equations, we solve step (b) using the closed form \eqref{Sec3_closedform}.
The \texttt{POT} package \citep{pot} in \texttt{Python} is used to solve the discrete optimal transport problem.
Throughout we fix $n_{1} = 0.4n$ and $p = 6$ for TWI. 
In addition, we apply the $k$-TWI (Algorithm \ref{alg:kWI}) with three cut-off points $n_{1} \in \{0.25n, 0.5n, 0.75n\}$. 
For the initializations of TWI, both linear interpolation and Kalman smoothing are used in our experiments. 
Kalman smoothing and gap filling with iterative SSA (iSSA henceforth) are implemented, respectively, via the \texttt{imputeTS} \citep{imputets} and \texttt{Rssa} \citep{golyandina2018} packages in \texttt{R}.
As discussed in Section \ref{Subsec::relatedworks}, 
the window length $L$ for which iSSA can be implemented is restricted. 
Here we use $L = 6$ for iSSA. 
In the reconstruction stage for iSSA, the first three eigentriples are used for grouping with the linear interpolation as initialization. 
The scalar filter algorithm (denoted as ScalarF henceforth) utilizes intervention analysis, which is detailed 
in Section S2 of the supplementary material.
Codes for the numerical experiments in this and the next section can be found in the author's Github page: \url{https://github.com/shuochieh/Wasserstein_imputation}.

\subsection{Linear and nonlinear univariate time series}

We first consider the following data-generating processes (DGPs). 
In the sequel, $\{\epsilon_{t}\}$ and $\{\epsilon_{t,j}\}$, for $j = 1,2,\ldots$, denote i.i.d.~sequences of standard Gaussian variables. 

\begin{model}[AR] \label{Sec5_modelAR}
    $x_{t} = \phi x_{t-1} + \epsilon_{t}$, where $\phi$ is set to 0.8.
\end{model}

\begin{model}[ARMA] \label{Sec5_modelARMA}
    $(1 - \phi_{1}B)x_{t} = (1 + \phi_{2}B)\epsilon_{t}$, where $B$ is the backshift operator and $(\phi_{1}, \phi_{2}) = (0.8, -0.6)$. 
\end{model}

\begin{model}[TAR; \citealp{tong1983}] \label{Sec5_modelTAR} 
If $x_{t-1} \leq \tau$, then $x_{t} = \phi_{1} x_{t-1} + \epsilon_{t}$. Otherwise, $x_{t} = \phi_{2} x_{t-1} + 0.5\epsilon_{t}$.
    The parameters $(\phi_{1}, \phi_{2}, \tau) = (-2, 0.7, 1)$.
\end{model}

\begin{model}[I(1) process] \label{Sec5_modelI1}
    Data $\{x_{t}\}$ are generated according to $x_{t} = x_{t-1} + z_{t} + \epsilon_{t,2}$, where $z_{t} = \phi z_{t-1} + 0.5\epsilon_{t,1}$ with $\phi = -0.7$.
\end{model}

\begin{model}[Cyclic series with noise (CYC)] \label{Sec5_modelCYC}
        $x_{t} = 10\cos(t(0.23\pi)) + 6 \cos(t(0.17\pi)) + 0.5\epsilon_{t}$.
\end{model}

DGPs \ref{Sec5_modelAR} and \ref{Sec5_modelARMA} are well-known stochastic time series models.
Note that DGP \ref{Sec5_modelARMA} admits an AR($\infty$) representation, so any finite-dimensional marginal distribution 
is inadequate to fully characterize the distribution of the underlying process.
DGP \ref{Sec5_modelTAR} is a celebrated threshold AR model which has found many applications in modeling nonlinear dynamics in the econometrics and statistics literature \citep{Montgomery1998, Li2012, tsay2018}.
DGP \ref{Sec5_modelI1} generates an I(1) series which typically produces a highly persistent stochastic trend.
Finally, DGP \ref{Sec5_modelCYC} features a deterministic cyclic trend with non-standard periods.
Specifically, it includes two sinusoidal components with periods of 8.70 and 11.77, respectively. 

Three measures of performance are considered.
First, we report the Wasserstein distance between the empirical 3-dimensional marginal distributions of the imputed series and the full data. 
Specifically, let $\{\hat{w}_{t}^{(i)}\}$ and $\{\tilde{x}_{t}^{(i)}\}$ be the imputed series and the original series in simulation $i$, respectively, for $i=1,2,\ldots,1000$. We report
\begin{align} \label{Sec5_wassloss}
    \frac{1}{1000}\sum_{i=1}^{1000}\mathcal{W}_{2} \left( \frac{1}{n-3} \sum_{t=2}^{n-1} \delta_{(\hat{w}_{t}^{(i)}, \hat{w}_{t-1}^{(i)}, \hat{w}_{t-2}^{(i)})}, \frac{1}{n-3} \sum_{t=2}^{n-1} \delta_{(\tilde{x}_{t}^{(i)}, \tilde{x}_{t-1}^{(i)}, \tilde{x}_{t-2}^{(i)})} \right),
\end{align}
as a measure of how well the imputed series approximates the original series.
Second, for DGP \ref{Sec5_modelAR}--\ref{Sec5_modelI1}, the model parameters are estimated using the imputed series, and the estimation root mean square errors (RMSEs) are reported. 
Third, we compare the RMSEs of the estimated autocovariance functions (ACFs). 
However, due to space constraints, the results for the ACFs are moved to Section S3.1 in the supplementary material.


The Wasserstein distance between the imputed and the original marginal distributions is reported in Table \ref{tab:Wassd_1}.
Across all five DGPs employed, the smallest loss is consistently achieved by TWI-type methods, indicating that TWI effectively approximates the distributional properties of the underlying series.
While other methods may yield good imputations for some specific DGPs, they are not consistent across all processes.
For instance, linear interpolations show low Wasserstein loss for DGP \ref{Sec5_modelAR} (AR), but it is less competitive when applied to other models. 
Similarly, Kalman smoothing and iSSA work well for DGP \ref{Sec5_modelCYC} (CYC), but both yield large Wasserstein loss for DGP \ref{Sec5_modelTAR} (TAR). 
In contrast, TWI methods, carefully initialized, consistently deliver solid performance across models and missing patterns. 
For the nonstationary Model \ref{Sec5_modelI1} (I(1)), Table \ref{tab:Wassd_1} shows the Wasserstein distance between the marginal distributions of the differenced series $\{\nabla \hat{w}_{t} \}$ and $\{ \nabla x_{t}\}$.
TWI again yields the lowest Wasserstein loss. 

Table \ref{tab:Wassd_1} also shows TWI often significantly improves the Wasserstein loss from its initialization. 
For example, for DGP \ref{Sec5_modelTAR} (TAR), TWI reduces the Wasserstein loss from Kalman smoothing initialization by approximately 45\% under missing pattern I (from 1.38 to 0.74).
Similarly, for DGP \ref{Sec5_modelCYC} (CYC), TWI reduces the Wasserstein loss by more than 50\% from linear initialization (from 1.96 to 0.79).  
Figure \ref{fig:TAR1_scatter} plots $\hat{w}_{t}$ against $\hat{w}_{t-1}$ when either or both are imputed under DGP \ref{Sec5_modelTAR}, alongside $(\tilde{x}_{t}, \tilde{x}_{t-1})$ of the original data. 
While Kalman smoothing introduces substantial distortion in the marginal distribution, TWI, initialized by Kalman smoothing, significantly corrects this bias, producing imputations that capture the nonlinearity in the data generating process.
Note that Kalman smoothing is based on an ARIMA state-space representation; hence the model is misspecified.
Nevertheless, it already serves as a good initialization for TWI.

\spacingset{1}
\begin{table}[t]
\centering
\caption{Wasserstein distance between the empirical marginal distributions of the imputed series and those of the full data, averaged over 1000 simulations. iSSA denotes gap filling by the iterative singular spectrum analysis. TWI$_\mathrm{lin}$ denotes temporal Wasserstein imputation using linear interpolation as initialization and TWI$_\mathrm{Kal}$ denotes temporal Wasserstein imputation with Kalman smoothing as initialization. Similar notations are used for $k$-TWI. For DGP \ref{Sec5_modelI1} (I(1)), the loss is computed using the first-differenced series. The smallest value for each model and missing pattern is marked in boldface.}
\begin{tabular}{@{}crrrrrrrr@{}}
\toprule
Model & Linear & iSSA & Kalman & ScalarF & TWI$_\mathrm{lin}$ & $k$-TWI$_\mathrm{lin}$ & TWI$_\mathrm{Kal}$ & $k$-TWI$_\mathrm{Kal}$        \\ \midrule
\multicolumn{9}{c}{{\bf Missing pattern I}} \\
AR      & 0.41 & {\bf 0.40} & 0.42 & 0.50 & {\bf 0.40} & 0.44 & 0.41 & 0.44 \\
ARMA    & 0.48 & 0.51 & 0.47 & 0.48 & 0.40 & {\bf 0.38} & 0.40 & {\bf 0.38} \\
TAR     & 1.12 & 1.29 & 1.38 & 1.13 & 0.96 & 0.81 & 0.74 & {\bf 0.63} \\
I(1)    & 0.71 & 0.68 & 0.62 & 4.38 & 0.58 & 0.53 & 0.54 & {\bf 0.50} \\
CYC     & 1.96 & 0.61 & 0.77 & 1.97 & 0.79 & 0.77 & {\bf 0.60} & 0.70 \\
NLVAR   & 2.98 & 3.04 & 3.01 & 2.94 & 2.25 & 2.14 & 2.19 & {\bf 2.12} \\
AL {\footnotesize ($\times 10$)}   & 0.70 & 0.85 & 4.17 & 0.55 & 0.43 & 0.37 & 0.38 & {\bf 0.35}\smallskip \\
\multicolumn{9}{c}{{\bf Missing pattern II}} \\
AR      & 0.44 & 0.43 & 0.50 & 0.76 & {\bf 0.39} & 0.44 & 0.43 & 0.45 \\
ARMA    & 0.47 & 0.69 & 0.51 & 0.53 & 0.36 & {\bf 0.34} & 0.40 & 0.35 \\
TAR     & 1.04 & 1.58 & 1.25 & 0.99 & 0.84 & 0.73 & 0.76 & {\bf 0.61} \\
I(1)    & 0.67 & 0.60 & 0.60 & 8.35 & 0.51 & 0.45 & 0.49 & {\bf 0.44} \\
CYC     & 2.58 & 1.47 & 0.77 & 3.79 & 2.62 & 1.60 & {\bf 0.62} & 0.72 \\
NLVAR   & 3.00 & 3.17 & 3.33 & 3.67 & 2.13 & {\bf 2.02} & 2.31 & 2.04 \\
AL {\footnotesize ($\times 10$)}   & 0.56 & 1.04 & 3.81 & 0.46 & 0.36 & 0.33 & 0.40 & {\bf 0.34} \\
\bottomrule
\end{tabular}
\label{tab:Wassd_1}
\end{table}
\spacingset{2}

\spacingset{1}
\begin{figure}
    \centering
    \begin{subfigure}[t]{0.48\textwidth}
        \centering
        \includegraphics[width=\linewidth, height=0.3\textheight]{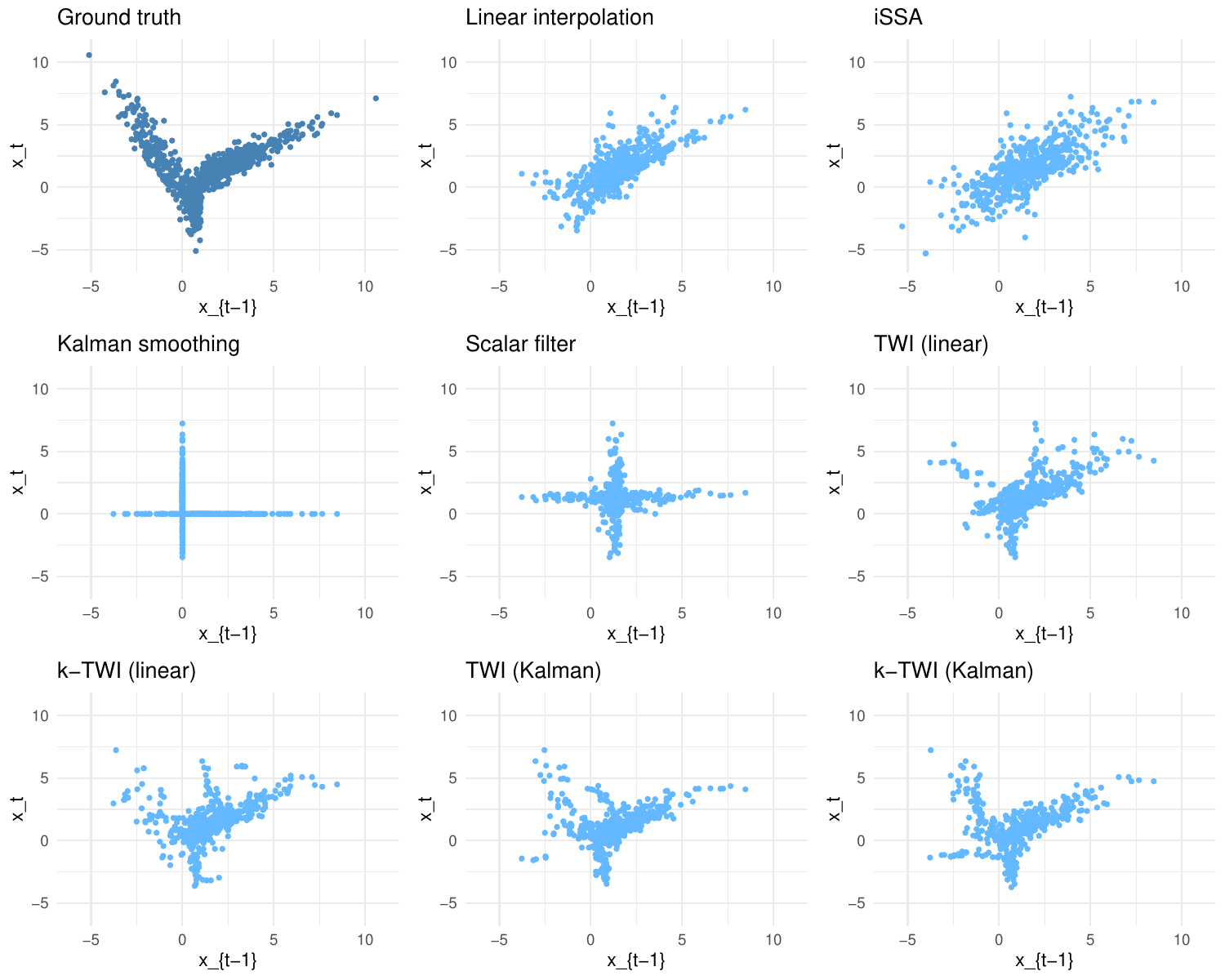}
        \caption{Under Missing pattern I}
    \end{subfigure} 
    \begin{subfigure}[t]{0.48\textwidth}
        \centering
        \includegraphics[width=\linewidth, height=0.3\textheight]{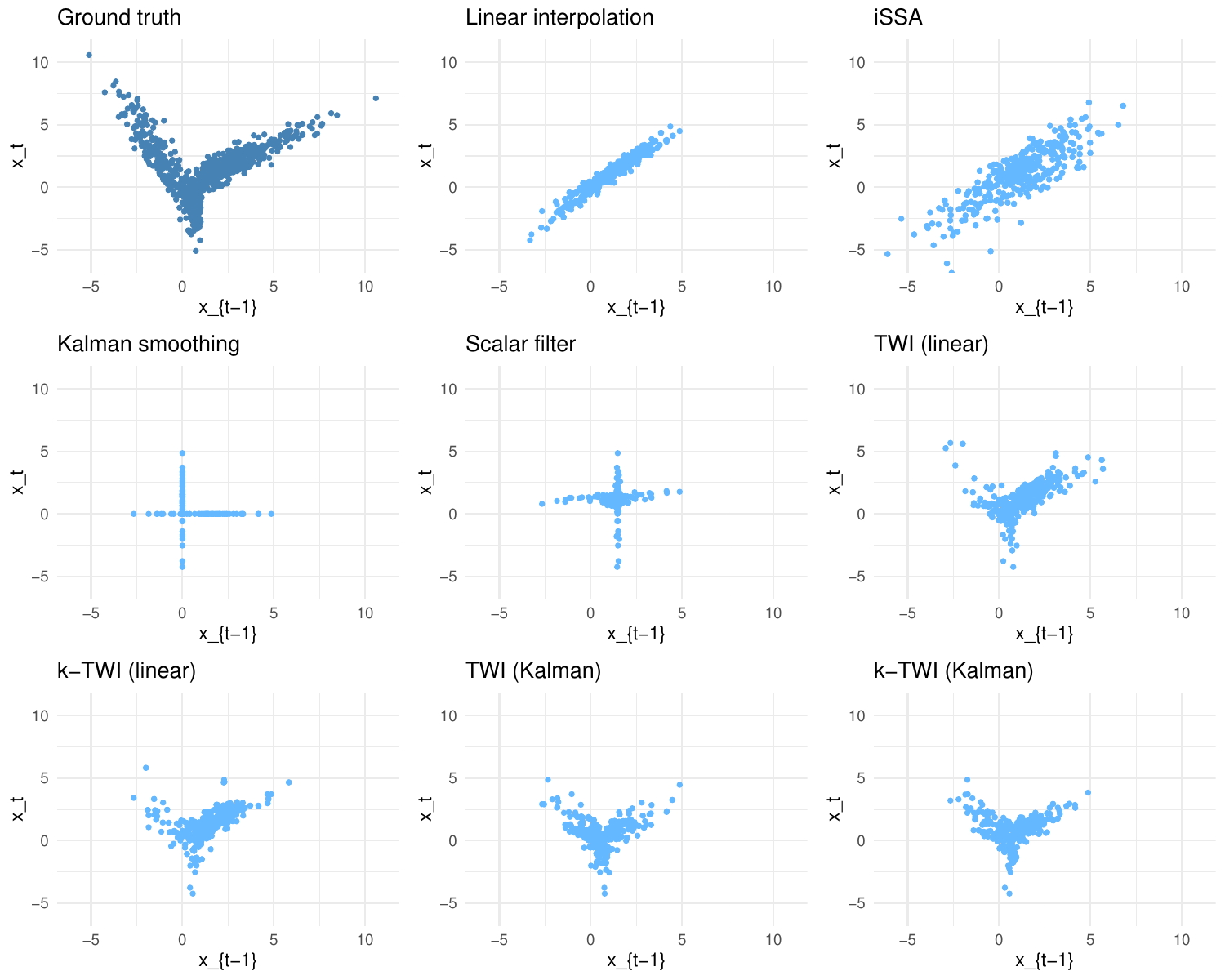}
        \caption{Under Missing pattern II}
    \end{subfigure} \\
    \caption{Scatter plots of $(x_{t-1}, x_{t})$ of the original data and $(\hat{w}_{t-1}, \hat{w}_{t})$ of the imputed series, when the data are generated from Model \ref{Sec5_modelTAR} (TAR).}
    \label{fig:TAR1_scatter}
\end{figure}
\spacingset{2}

Next, we turn to the model parameters estimated from the imputed series. Table \ref{tab:models_RMSE} presents the RMSEs for each model parameter.
For DGP \ref{Sec5_modelAR} (AR) and DGP \ref{Sec5_modelARMA} (ARMA), Kalman smoothing and the scalar filter yield satisfactory estimates among the benchmarks, with TWI offering only limited improvements. 
For the nonlinear DGP \ref{Sec5_modelTAR} (TAR), the results are drastically different.
TWI produces the most accurate parameter estimates and can achieve more than 50\% reduction in estimation errors compared to linear interpolation and Kalman smoothing.
These estimation results demonstrate TWI's strength in learning the underlying nonlinear dynamics. 
For DGP \ref{Sec5_modelI1} (I(1)), TWI either outperforms or stays on par with the best benchmark, with Kalman smoothing being a strong initialization.

\spacingset{1}
\begin{table}[h]
\centering
\caption{Estimation root mean squared errors for parameter estimation errors. The results are averaged over 1000 simulations. See Table \ref{tab:Wassd_1} for notations.}
\begin{tabular}{@{}cccccccccc@{}}
\toprule
Model & & Linear & iSSA & Kalman & ScalarF & TWI$_\mathrm{lin}$ & $k$-TWI$_\mathrm{lin}$ & TWI$_\mathrm{Kal}$ & $k$-TWI$_\mathrm{Kal}$        \\ \midrule
\multicolumn{10}{c}{\bf Missing pattern I} \\
AR & $\phi$ & 0.06 & 0.06 & 0.05 & {\bf 0.03} & {\bf 0.03} & 0.04 & {\bf 0.03} & 0.04\smallskip \\
ARMA & $\phi_{1}$ & 0.19 & 0.28 & 0.12 & {\bf 0.08} & 0.17 & 0.16 & 0.11 & 0.11 \\
     & $\phi_{2}$ & 0.46 & 0.63 & 0.18 & {\bf 0.10} & 0.33 & 0.26 & 0.16 & 0.15\smallskip \\
TAR  & $\phi_{1}$ & 1.02 & 1.41 & 0.82 & 0.63 & 0.71 & 0.58 & 0.53 & {\bf 0.44} \\
     & $\phi_{2}$ & 0.07 & 0.07 & 0.20 & 0.04 & 0.05 & 0.03 & 0.03 & {\bf 0.02} \\
     & $\tau$     & {\bf 0.01} & 0.26 & 0.35 & 0.20 & 0.05 & 0.05 & 0.03 & 0.03\smallskip \\
I(1) & $\phi$     & 0.96 & 1.03 & 0.45 & 0.68 & 0.48 & 0.32 & 0.30 & {\bf 0.24}\smallskip \\ 
NLVAR& $\phi_{11}$& 0.19 & 0.11 & 0.15 & {\bf 0.03} & 0.06 & {\bf 0.03} & 0.05 & {\bf 0.03} \\
     & $\phi_{12}$& 3.10 & 2.43 & 2.87 & 1.67 & 1.34 & 1.02 & 1.22 & {\bf 0.97} \\
     & $\phi_{21}$& {\bf 0.03} & 0.07 & {\bf 0.03} & {\bf 0.03} & {\bf 0.03} & {\bf 0.03} & {\bf 0.03} & {\bf 0.03} \\
     & $\phi_{22}$& 0.18 & 0.09 & 0.10 & {\bf 0.04} & 0.05 & {\bf 0.04} & {\bf 0.04} & {\bf 0.04}\smallskip \\
\multicolumn{10}{c}{\bf Missing pattern II} \\
AR & $\phi$ & 0.05 & 0.06 & 0.04 & 0.04 & 0.03 & {\bf 0.02} & 0.03 & {\bf 0.02}\smallskip \\
ARMA & $\phi_{1}$ & {\bf 0.03} & 0.11 & 0.09 & 0.08 & 0.05 & 0.06 & 0.09 & 0.09 \\
     & $\phi_{2}$ & 0.16 & 0.46 & 0.11 & 0.{\bf 10} & 0.12 & 0.11 & 0.11 & 0.12\smallskip \\
TAR  & $\phi_{1}$ & 0.72 & 1.66 & 0.38 & 0.25 & 0.25 & 0.20 & 0.29 & 0.21 \\
     & $\phi_{2}$ & 0.09 & 0.10 & 0.05 & 0.03 & 0.06 & 0.05 & 0.02 & {\bf 0.01} \\
     & $\tau$     & {\bf 0.01} & 0.34 & {\bf 0.01} & {\bf 0.01} & {\bf 0.01} & {\bf 0.01} & {\bf 0.01} & {\bf 0.01}\smallskip \\
I(1) & $\phi$     & {\bf 0.08} & 0.53 & 0.11 & 0.70 & 0.09 & 0.10 & 0.10 & 0.10\smallskip \\
NLVAR& $\phi_{11}$& 0.19 & 0.04 & 0.08 & 0.02 & 0.02 & {\bf 0.01} & 0.02 & {\bf 0.01} \\
     & $\phi_{12}$& 2.67 & 1.49 & 1.68 & 1.26 & 0.61 & 0.46 & 0.61 & {\bf 0.43} \\
     & $\phi_{21}$& {\bf 0.02} & 0.09 & 0.03 & 0.03 & 0.06 & 0.05 & 0.04 & 0.03 \\
     & $\phi_{22}$& 0.15 & 0.11 & 0.05 & {\bf 0.04} & {\bf 0.04} & {\bf 0.04} & {\bf 0.04} & {\bf 0.04} \\
\bottomrule
\end{tabular}
\label{tab:models_RMSE}
\end{table}
\spacingset{2}

\subsection{Multivariate time series}
We consider next models for multivariate nonlinear time series.

\begin{model}[Nonlinear VAR; NLVAR] \label{Sec5_modelNLVAR}
Let $\mathfrak{s}(z) = 1 / (1 + \exp(-z)) - 0.5$ be the centered sigmoid function. The data are generated as 
    \begin{align*}
        x_{t,1} = \phi_{11} x_{t-1,1} + \phi_{12} \mathfrak{s}(3x_{t-1,2}) + 0.25 \epsilon_{t,1}, \quad
        x_{t,2} = \phi_{21} x_{t-1,1} + \phi_{22} x_{t-1,2} + 3\epsilon_{t,2},
    \end{align*}
    where $(\phi_{11}, \phi_{12}, \phi_{21}, \phi_{22}) = (0.3, 8, 0, 0.4)$. 
\end{model}

DGP \ref{Sec5_modelNLVAR} is a VAR model with a nonlinear term.
In particular, the impact of $x_{t-1,2}$ on $x_{t,1}$ is nonlinear and is modeled by a sigmoid function.
As depicted in Figure \ref{fig:NLVAR_scatter}, TWI successfully captures the underlying nonlinear dynamics, whereas other imputation methods fail to preserve meaningful relationships.
This is validated by the results shown in Table \ref{tab:Wassd_1}, where TWI achieves the lowest Wasserstein distance.
Moreover, as shown in Table \ref{tab:models_RMSE}, TWI provides considerable error reduction in the parameter associated with the nonlinear effect, $\phi_{12}$, compared to the benchmarks.
It attains more than 50\% reduction in RMSE compared to the scalar filter, the best-performing benchmark.

We also consider a challenging multivariate nonlinear model, which generates the so-called compositional time series.
The data are generated from the additive logistic (AL) model of \cite{brunsdon1998}.
Due to space constraints, we present the simulation results and discussion in Section S3.2 of the supplementary material.

\spacingset{1}
\begin{figure}[h!]
    \centering
    \begin{subfigure}[t]{0.49\textwidth}
        \includegraphics[width=\linewidth]{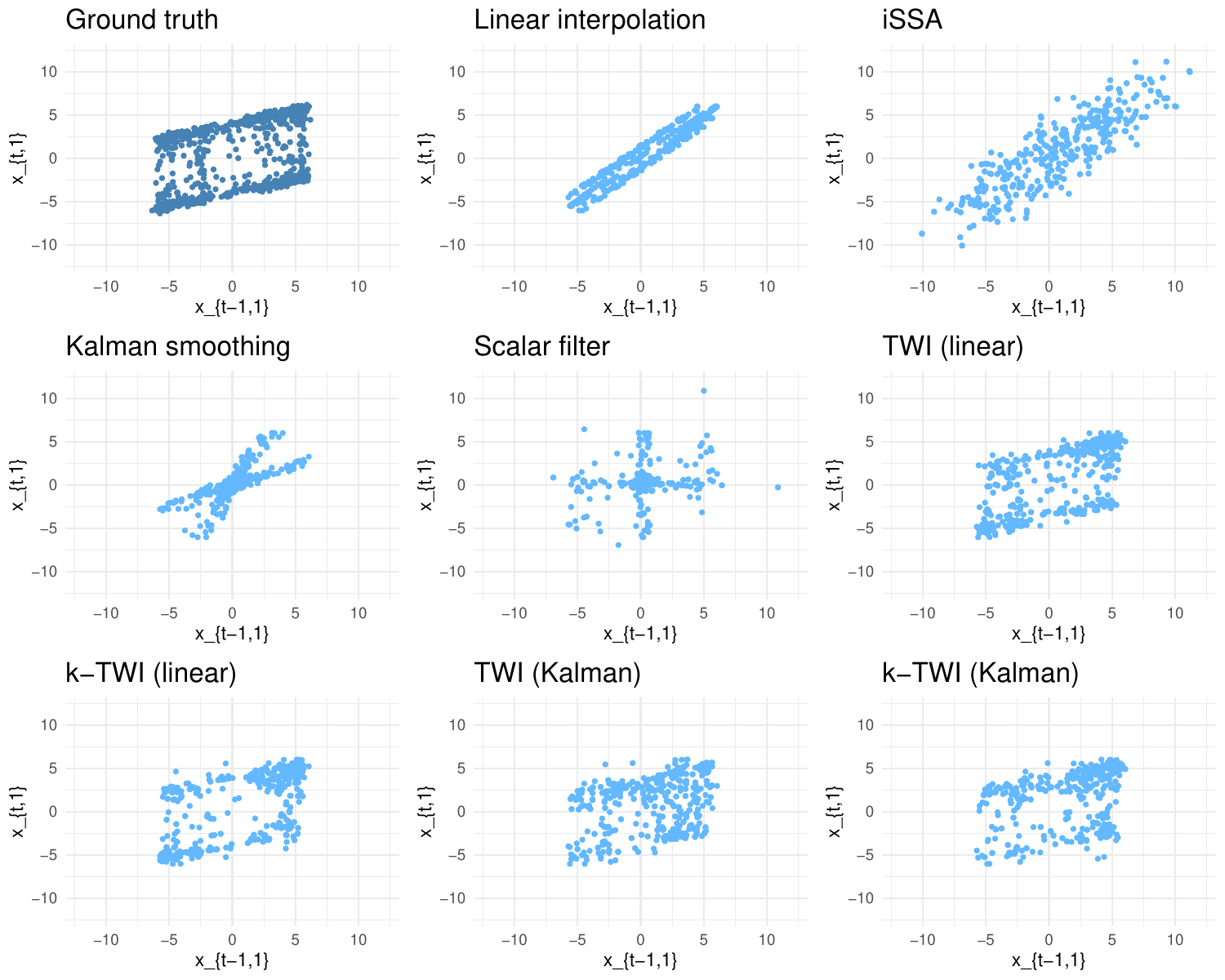}
        \caption{$\hat{w}_{t,1}$ on $\hat{w}_{t-1, 1}$}
    \end{subfigure}
    \begin{subfigure}[t]{0.49\textwidth}
        \includegraphics[width=\linewidth]{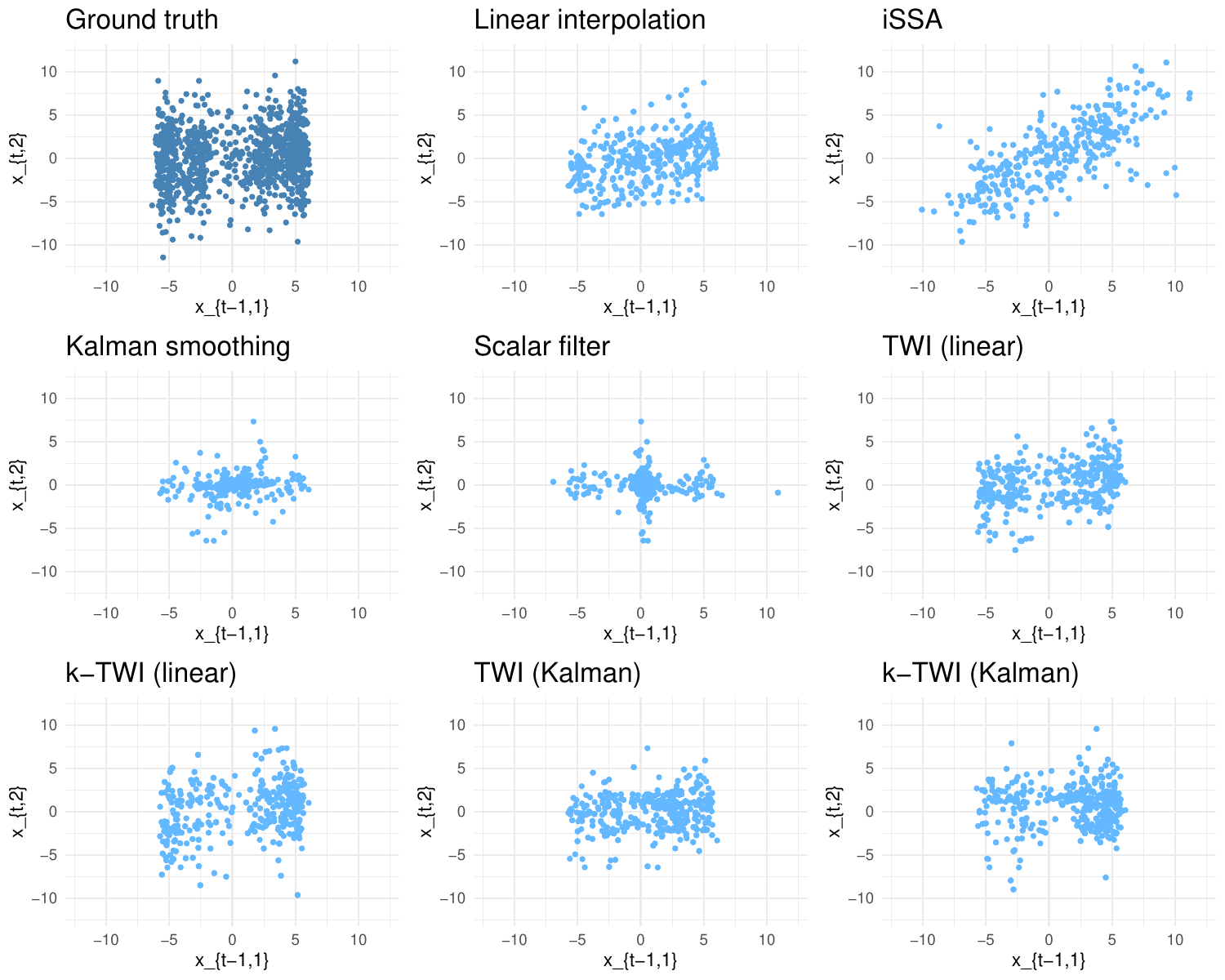}
        \caption{$\hat{w}_{t,2}$ on $\hat{w}_{t-1, 1}$}
    \end{subfigure} \\
    \begin{subfigure}[t]{0.49\textwidth}
        \includegraphics[width=\linewidth]{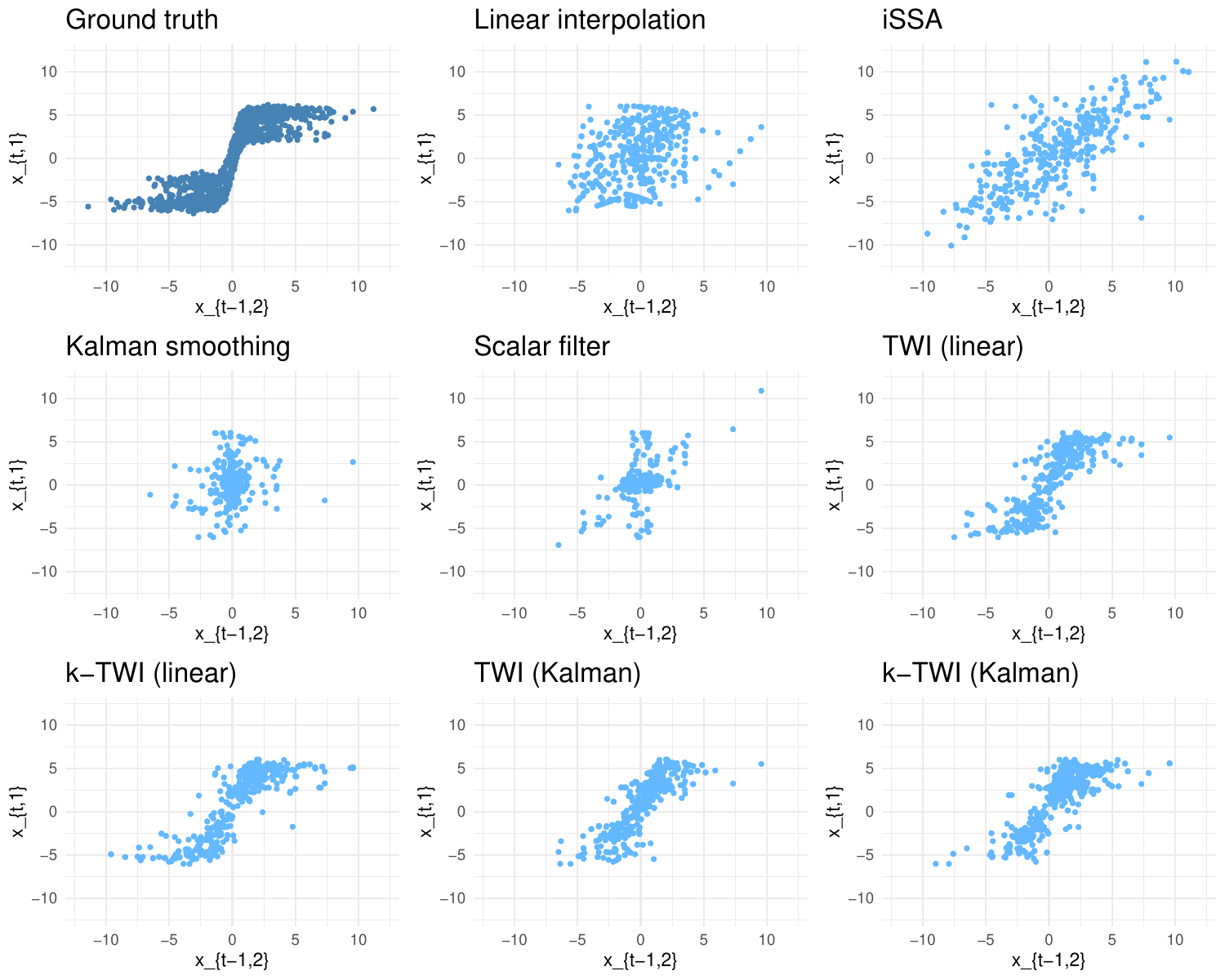}
        \caption{$\hat{w}_{t,1}$ on $\hat{w}_{t-1, 2}$}
    \end{subfigure}
    \begin{subfigure}[t]{0.49\textwidth}
        \includegraphics[width=\linewidth]{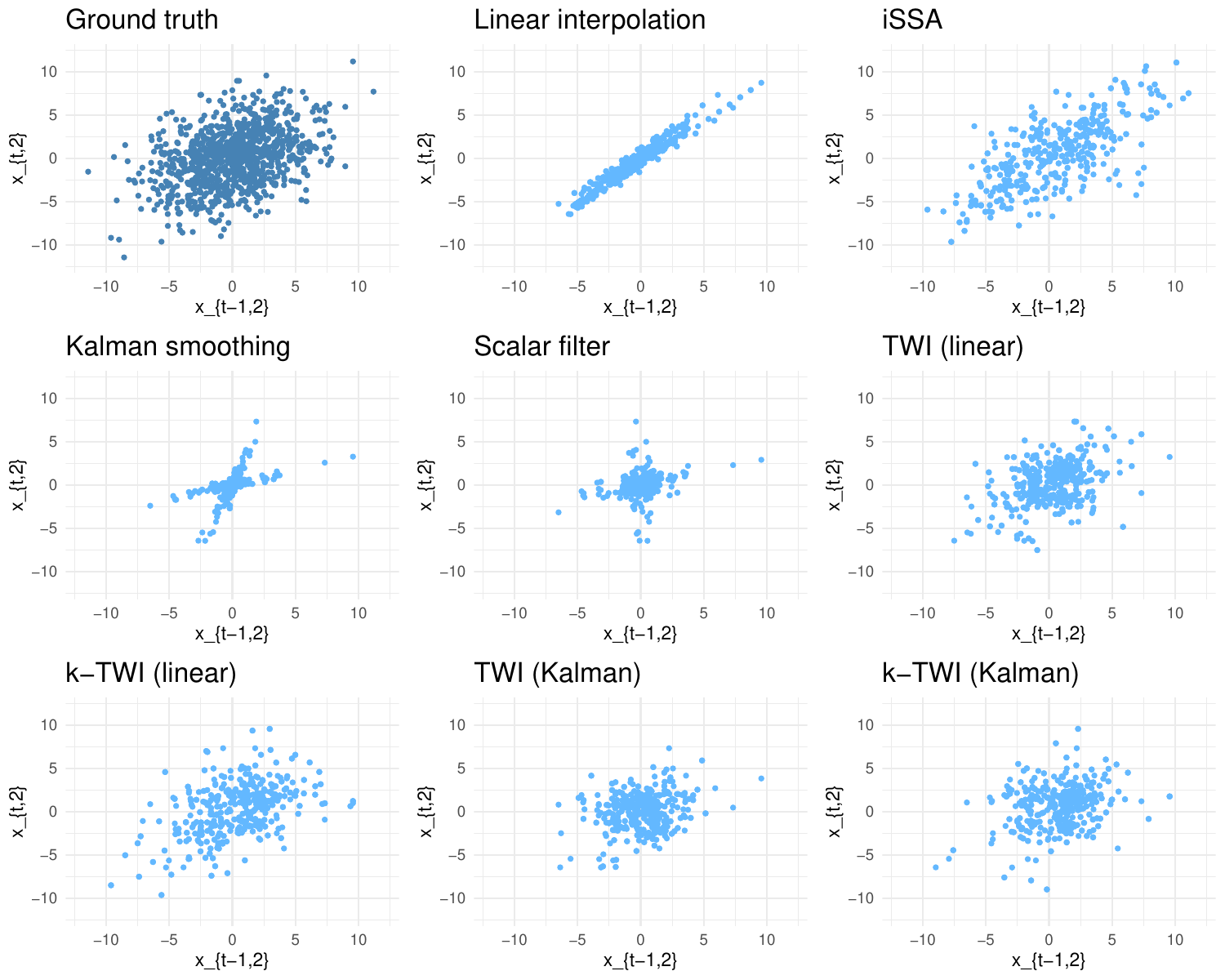}
        \caption{$\hat{w}_{t,2}$ on $\hat{w}_{t-1, 2}$}
    \end{subfigure}
    \caption{Scatter plots of $(x_{t-1,i}, x_{t,j})$ of the original data and $(\hat{w}_{t-1,i}, \hat{w}_{t,j})$ of the imputed series for $i,j \in \{1,2\}$, when the data are generated from Model \ref{Sec5_modelNLVAR} under missing pattern II.}
    \label{fig:NLVAR_scatter}
\end{figure}
\spacingset{2}

\section{Groundwater data application} \label{Sec::realdata}

In this section, we apply the proposed TWI to a Taiwan groundwater dataset, which was collected and analyzed by \citet{Hsu2020}.
Because of the island's steep topography and uneven rainfall distribution, groundwater management is of considerable importance, particularly as it is a critical water source.
However, the statistical analysis of groundwater levels is plagued by frequent missing data, especially from deeper wells. 

To study the seasonal groundwater variations, we first transform the hourly data between October 1992 and August 2020 into monthly records. 
For each site, monthly groundwater level is calculated as the average in a given month, with months classified as missing if more than one-third of the observations are missing.
There are many outlying values in the series, including abrupt changes in groundwater levels that are more than 5 standard deviations away from the sample mean.
We also treat such extreme values as missing data.
For demonstration, our analysis focuses on the sites with less than 10 missing values, reducing the sample to 15 sites.
As shown in Figure \ref{fig:GW_timeplot}, these series exhibit seasonality and structural changes, which are quite representative of the dataset and are common real-world features.

\spacingset{1}
\begin{figure}[t]
    \centering
    \includegraphics[width=\textwidth, height = 0.35\textheight]{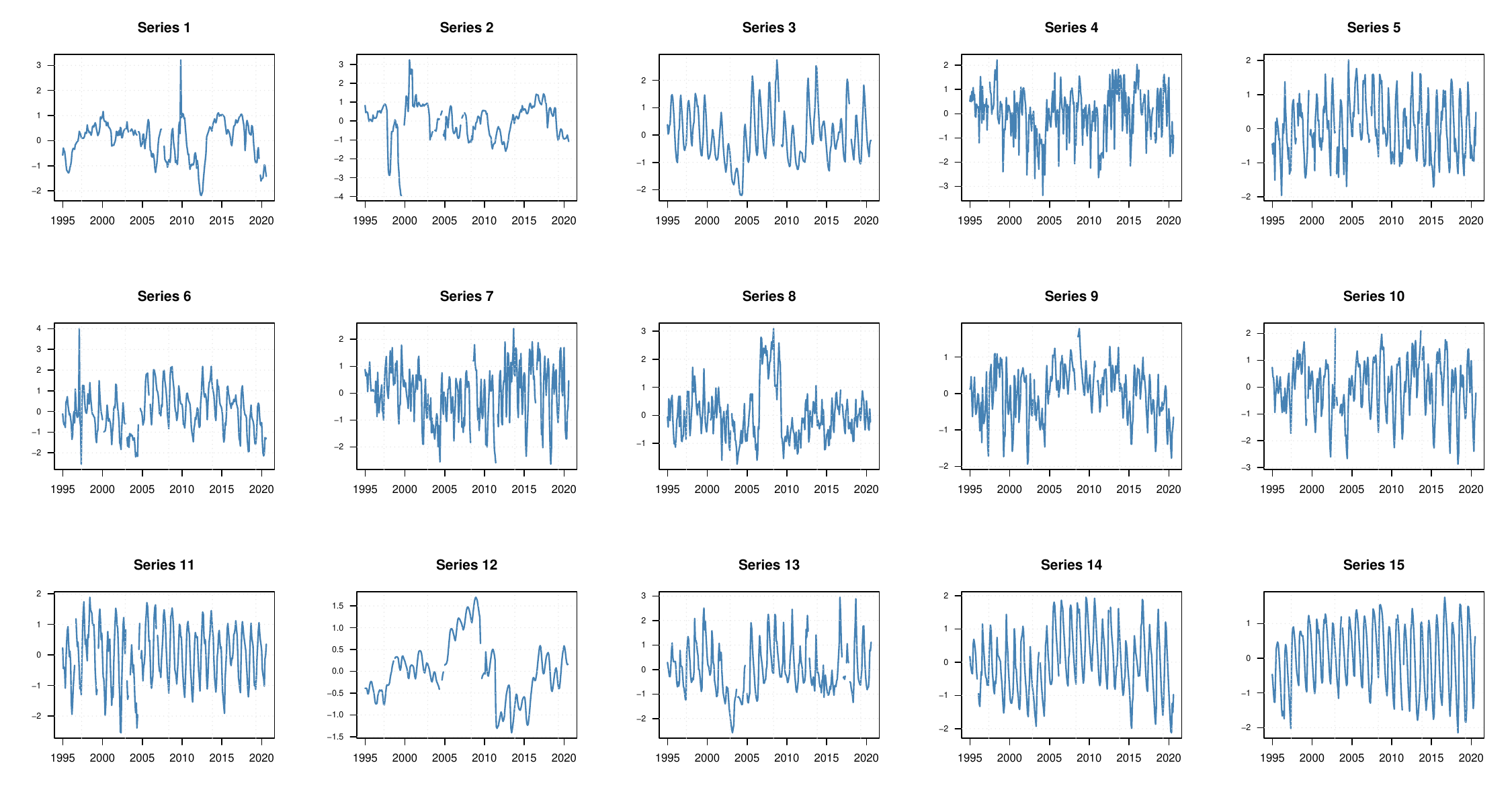}
    \caption{Time plots of the groundwater levels of the 15 sites under study. Each series is normalized by subtracting its mean and dividing by its standard deviation.}
    \label{fig:GW_timeplot}
\end{figure}
\spacingset{2}

In the following, we randomly omit $\lfloor \rho n \rfloor$ observations to create missing data, where $\rho \in \{0.1, 0.2, 0.3, 0.4, 0.5, 0.6\}$ is the missing ratio, $n = 308$ is the number of observations (including missing values) in the period under study, and $\lfloor x \rfloor$ is the largest integer $\leq x$ for $x > 0$.
However, the last 36 observations of each series are left intact because otherwise, the SSA could not be implemented due to lack of intact columns in the trajectory matrix. 
Since TWI performs best with stationary series, we first use kernel smoothing (estimated with missing data) to remove the trends, and then apply ($k$-)TWI to the residuals.
In addition to the methods employed in Section \ref{Sec::Simulation}, we also compare with imputing the missing values using the sample mean. 
Due to space constraints, the implementation details of the methods are provided in Section S3.3 of the supplementary material.




We consider two downstream tasks.
First, the imputed series are used to estimate some summary statistics, including the autocovariances, autocorrelations, and partial autocorrelations. 
Specifically, let $\gamma_{j}(h)$ be the statistic at the $h$-th lag obtained using the $j$-th series (without the artificial missing data), and $\hat{\gamma}_{j}(h)$ be the corresponding estimates from the imputed series. 
Table \ref{Realdata:tab1} reports the mean absolute error (MAE), $\frac{1}{45}\sum_{j=1}^{15}\sum_{h=1}^{3}|\hat{\gamma}_{j}(h)-\gamma_{j}(h)|$, for the autocovariances. 
When the missing ratio is mild, for instance $\rho \leq 0.3$, iSSA, Kalman smoothing, and linear interpolation perform reasonable well, while TWI and $k$-TWI are quite close to the best-performing methods.
Mean imputation significantly distorts the estimates. 
However, when $\rho > 0.3$, TWI and $k$-TWI show some advantages.
This is expected because when the missing ratio is high, accurate estimation of the missing data is difficult, but distribution matching is still possible.
Table \ref{Realdata:tab2} shows the MAE for estimating the autocorrelations.
For $\rho \leq 0.2$, most methods yield similar estimates. 
When $\rho > 0.2$, $k$-TWI outperforms the competing methods, among which the Kalman smoothing is a strong contender.
Meanwhile, TWI is somewhat close to Kalman smoothing for $\rho \leq 0.4$. 
Linearly interpolated series, however, has autocorrelations quite different from the original series.
Finally, Table \ref{Realdata:tab3} shows the MAE for the partial autocorrelations.
For $\rho \geq 0.3$, TWI and $k$-TWI consistently show good performance. 
The MAE of Kalman smoothing is similar to the proposed methods only when $\rho \leq 0.4$.
Unlike in previous cases, the partial autocorrelations of the iSSA imputations are inaccurate, showing that it often neglects stochastic dynamic relationships.

\spacingset{1}
\begin{table}[t]
\centering
\caption{Mean absolute errors ($\times 10$) for estimating autocovariances. Mean stands for imputation by the sample mean. For each missing ratio $\rho$, the smallest two values are in boldface.}
\label{Realdata:tab1}
\begin{tabular}{@{}lccccccc@{}}
\toprule
$\rho$ & Linear & iSSA & Kalman & ScalarF & Mean & TWI & $k$-TWI \\ \midrule
0.1 & {\bf 0.14} & 0.16 & {\bf 0.14} & 0.31 & 0.98 & 0.16 & 0.18 \\
0.2 & {\bf 0.21} & {\bf 0.20} & 0.38 & 0.72 & 1.91 & 0.27 & 0.34 \\
0.3 & 0.41 & {\bf 0.31} & {\bf 0.30} & 1.01 & 2.46 & 0.35 & 0.33 \\
0.4 & 0.50 & 0.49 & {\bf 0.43} & 1.58 & 3.10 & {\bf 0.46} & {\bf 0.46} \\
0.5 & 0.77 & 0.63 & 0.60 & 2.35 & 3.54 & {\bf 0.56} & {\bf 0.52} \\
0.6 & 0.98 & {\bf 0.91} & 1.25 & 2.90 & 3.84 & 0.99 & {\bf 0.79} \\ \bottomrule
\end{tabular}
\end{table}
\spacingset{2}

\spacingset{1}
\begin{table}[t]
\centering
\caption{Mean absolute errors ($\times 10$) for autocorrelations. See Table \ref{Realdata:tab1} for notations.}
\label{Realdata:tab2}
\begin{tabular}{@{}lccccccc@{}}
\toprule
$\rho$ & Linear & iSSA & Kalman & ScalarF & Mean & TWI & $k$-TWI \\ \midrule
0.1 & 0.18 & 0.17 & {\bf 0.12} & 0.15 & 0.63 & {\bf 0.13} & 0.14 \\
0.2 & 0.27 & {\bf 0.25} & 0.28 & 0.35 & 1.37 & {\bf 0.23} & 0.26 \\
0.3 & 0.55 & 0.40 & 0.36 & 0.51 & 1.76 & {\bf 0.31} & {\bf 0.28} \\
0.4 & 0.73 & 0.47 & {\bf 0.42} & 0.72 & 2.37 & 0.45 & {\bf 0.38} \\
0.5 & 1.17 & 0.72 & {\bf 0.56} & 1.22 & 2.85 & 0.65 & {\bf 0.44} \\
0.6 & 1.52 & {\bf 0.92} & 1.09 & 1.81 & 3.17 & 1.23 & {\bf 0.76} \\ \bottomrule
\end{tabular}
\end{table}
\spacingset{2}

\spacingset{1}
\begin{table}[t]
\centering
\caption{Mean absolute errors ($\times 10$) for the partial autocorrelations. See Table \ref{Realdata:tab1} for notations.}
\label{Realdata:tab3}
\begin{tabular}{@{}lccccccc@{}}
\toprule
$\rho$ & Linear & iSSA & Kalman & ScalarF & Mean & TWI & $k$-TWI \\ \midrule
0.1 & {\bf 1.10} & 1.13 & {\bf 1.11} & 1.34 & 2.27 & 1.14 & 1.27 \\
0.2 & {\bf 1.09} & 1.26 & {\bf 1.21} & 1.72 & 2.95 & 1.39 & 1.47 \\
0.3 & {\bf 1.26} & 1.60 & 1.50 & 1.91 & 3.19 & {\bf 1.34} & {\bf 1.34} \\
0.4 & {\bf 1.38} & 2.00 & 1.54 & 2.01 & 3.60 & {\bf 1.51} & 1.55 \\
0.5 & 1.48 & 1.76 & 1.88 & 2.24 & 3.81 & {\bf 1.47} & {\bf 1.46} \\
0.6 & 1.64 & 2.20 & 2.21 & 2.85 & 4.07 & {\bf 1.60} & {\bf 1.58} \\ \bottomrule
\end{tabular}
\end{table}
\spacingset{2}

Second, we use the first 200 observations from the imputed series to estimate an AR(15) model, where the estimated coefficients are then used in predicting, in a one-step-ahead rolling window fashion, the rest 108 observations from the true data.
Figure \ref{Realdata:fig1} plots the root mean squared prediction errors. 
For most series except Series 1 and 2, TWI and $k$-TWI are among the methods yielding the lowest prediction errors. 
Their advantage is again more pronounced when the missing ratio is large.
In addition, the TWI-type methods are stable as $\rho$ increases, whereas the iSSA and Kalman smoothing may suffer from large errors. 

In summary, the above analysis shows that the proposed TWI and $k$-TWI methods, through distribution matching, are useful in imputing missing values of time series for downstream statistical analysis. 
Compared to the existing nonparametric iSSA method and the parametric Kalman smoothing, such advantage is especially valuable when the missing ratio is high.

\spacingset{1}
\begin{figure}[t]
    \centering
    \includegraphics[width=\linewidth]{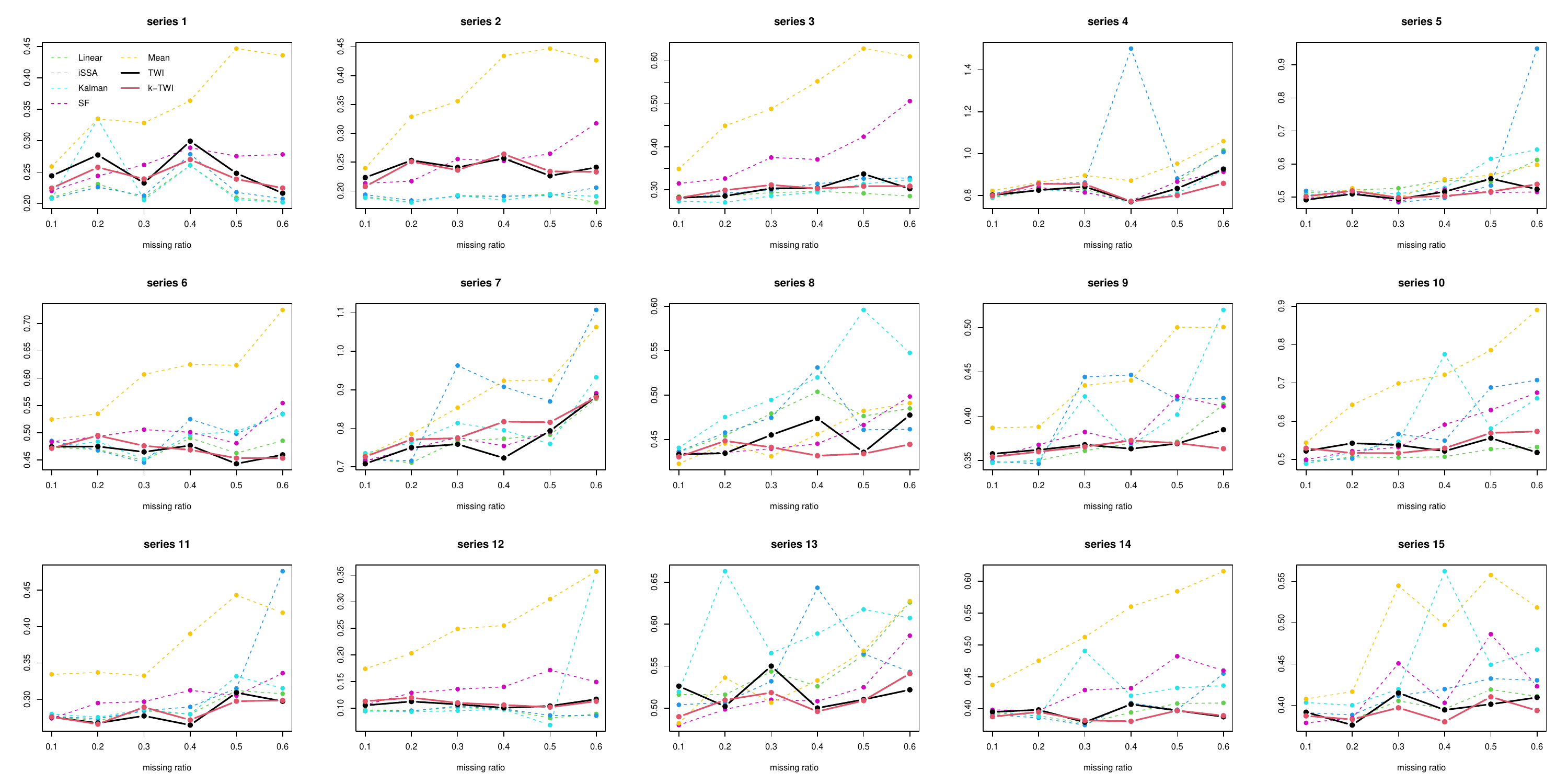}
    \caption{Root mean squared prediction errors for each series. The AR(15) model is trained using the first 200 observations of the imputed series, which is then tested on the rest 108 data from the uncontaminated true series (excluding missing values).}
    \label{Realdata:fig1}
\end{figure}
\spacingset{2}

\section{Conclusion}
This paper introduced a temporal Wasserstein imputation, a novel nonparametric method for imputing missing data in time series.
Compared to existing approaches, TWI (1) does not require any model specification prior to imputation; (2) can be carried out in a principled manner, and (3) yields imputations whose distributional properties are more aligned with those of the underlying process, making downstream analyses more reliable. 
In simulation study and real data analysis, TWI has shown its versatility and contributions to the existing literature, especially for the cases of high missing ratio and of imputing stochastic, nonlinear time series.

\bigskip
\begin{center}
{\large\bf SUPPLEMENTARY MATERIAL}
\end{center}

\begin{description}

\item[Supplementary material:] The file contains the proofs for the results presented in Section \ref{Sec::theory} and some technical results. It also documents the scalar filter algorithm of \citet{pena2021} and complementary numerical results, including the simulation results for autocovariance estimation and for a compositional time series DGP.
Finally, it includes the implementation details of the groundwater data analysis.

\end{description}

\bibliographystyle{apalike}
\bibliography{thebib}

\newpage

\spacingset{1.75} 

\setcounter{page}{1}
\setcounter{section}{0}
\setcounter{equation}{0}
\def\theequation{S\arabic{section}.\arabic{equation}}
\def\thetable{S\arabic{table}}
\def\thesection{S\arabic{section}}
\setcounter{table}{0}
\setcounter{figure}{0}
\renewcommand\thetheorem{S\arabic{theorem}}
\renewcommand\thelemma{S\arabic{lemma}}
\renewcommand\thefigure{S\arabic{figure}}

\begin{center}
    {\Large  \bf Supplementary material to ``Temporal Wasserstein Imputation: A Versatile Method for Time Series Imputation''} \\
    Shuo-Chieh Huang, Tengyuan Liang, and Ruey S. Tsay
\end{center}

This supplementary material contains three sections.
Section S1 collects the proofs of the propositions presented in main text as well as some technical results.
Section S2 documents the scalar filter algorithm of \citet{pena2021} for completeness.
Section S3 contains complementary numerical results, including the simulation results for autocovariance estimation and imputing data generated from a compositional time series model.
It also includes the details of implementation for the methods employed in the groundwater data analysis.

\section{Technical proofs} \label{App::Proof}

\begin{proof}[Proof of Proposition \textcolor{red}{3}]
    Clearly, (a) holds by construction of the algorithm. 
    Since Algorithm \textcolor{red}{1} 
    is an instance of the Gauss-Siedel algorithm, (b) follows from \cite{GRIPPO2000127}.
    Thus it remains to prove (c). 
    Optimizing $\mathbf{\Pi}$ while holding $\mathbf{w}$ fixed is the optimal transport problem, which is clearly convex. 
    Now fix $\mathbf{\Pi}$. Note that
    \begin{align*}
        F(\mathbf{w}, \mathbf{\Pi}) =& \sum_{i=p-1}^{n_{1}}\sum_{j=n_{1}+1}^{n-1} \pi_{ij} \sum_{h=0}^{p-1} \ell(w_{i-h} - w_{j-h}) + \frac{\lambda}{2}\Vert \mathbf{w} \Vert^{2} \\
        =& \sum_{h=0}^{p-1} \underbrace{\sum_{i=p-1}^{n_{1}}\sum_{j=n_{1}+1}^{n-1} \pi_{ij} \ell(w_{i-h} - w_{j-h})}_{F_{h}(\mathbf{w}, \mathbf{\Pi})} + \frac{\lambda}{2}\Vert \mathbf{w} \Vert^{2}. 
    \end{align*}
    It suffices to show each $F_{h}$ is convex in $\mathbf{w}$. 
    In the following, we only prove the case of $F_{0}$ since the arguments for the other cases are similar. 
    Let $\mathbf{u} = (w_{p-1}, w_{1}, \ldots, w_{n_{1}})^{\top}$ and $\mathbf{v} = (w_{n_{1}+1}, \ldots w_{n-1})^{\top}$. Then
    \begin{align*}
        \frac{\partial F_{0}}{\partial \mathbf{u}} &= \left( \sum_{j=n_{1}+1}^{n-1}\pi_{ij}\ell'(w_{i} - w_{j}): i = p-1, \ldots,n_{1} \right)^{\top},  \\
        \frac{\partial F_{0}}{\partial \mathbf{v}} &= \left( -\sum_{i=p-1}^{n_{1}}\pi_{ij}\ell'(w_{i} - w_{j}): j = n_{1}+1, \ldots,n-1 \right)^{\top}.
    \end{align*}
    Moreover,
    \begin{align}
        \frac{\partial^{2} F_{0}}{\partial \mathbf{u}^{2}} =& \mathrm{diag}\left( \sum_{j=n_{1}+1}^{n-1}\pi_{ij}\ell''(w_{i} - w_{j}): i = p-1,\ldots,n_{1} \right), \label{2der1} \\
        \frac{\partial^{2} F_{0}}{\partial \mathbf{v}^{2}} =& \mathrm{diag}\left( \sum_{i=p-1}^{n_{1}} \pi_{ij}\ell''(w_{i} - w_{j}): j = n_{1}+1, \ldots, n-1 \right), \label{2der2} \\
        \frac{\partial^{2} F_{0}}{ \partial \mathbf{u} \partial \mathbf{v}} =& -\left( \pi_{ij}\ell''(w_{i}-w_{j}) \right)_{p-1 \leq i \leq n_{1}, n_{1}+1 \leq j \leq n-1}. \label{2der3}
    \end{align}
    Since $\ell$ is convex, $\ell'' \geq 0$. From \eqref{2der1}--\eqref{2der3}, it is readily seen that the Hessian matrix $\mathbf{H} = \partial^{2} F_{0} / \partial \mathbf{w}^{2}$ is weakly diagonally dominant with nonnegative diagonal entries.
    It follows from the Gershgorin circle theorem that $\mathbf{H}$ is positive semi-definite. 
    Thus $F_{0}$ is convex in $\mathbf{w}$.
\end{proof}



In order to prove Proposition \textcolor{red}{4}, we need the following lemma.

\begin{lemma} \label{lem1}
    Suppose $\{\tilde{x}_{t}\}$ satisfies (A1) and (A2). Then
    \begin{align*}
        \mathcal{W}_{k} \left( \frac{1}{n-p+1}\sum_{t=p-1}^{n-1} \delta_{\mathbf{v}_{t}(\tilde{\mathbf{x}})}, \mu_{p} \right) \xrightarrow{n \rightarrow \infty} 0 \quad \mathrm{a.s.},
    \end{align*}
    where $\tilde{\mathbf{x}} = (\tilde{x}_{0}, \ldots, \tilde{x}_{n-1})$.
\end{lemma}
\begin{proof}
    By Theorem 7.12 of \cite{villani2021}, it suffices to show that with probability one,
    \begin{align} \label{lem1-cond1}
        \frac{1}{n-p+1}\sum_{t=p-1}^{n-1} \delta_{\mathbf{v}_{t}(\tilde{\mathbf{x}})} \Rightarrow \mu_{p},
    \end{align}
    and 
    \begin{align} \label{lem1-cond2}
        \frac{1}{n-p+1} \sum_{t=p-1}^{n-1} \Vert \mathbf{v}_{t}(\tilde{\mathbf{x}}) \Vert^{k} \rightarrow \mathbb{E}\Vert \mathbf{v}_{p-1}(\tilde{\mathbf{x}}) \Vert^{k}.
    \end{align} 
    By Theorem 36.4 of \cite{billingsley1995}, the processes $\{\Vert \mathbf{v}_{t}(\tilde{\mathbf{x}}) \Vert^{k}\}_{t}$ and $\{\exp(\mathrm{i}\mathbf{h}^{\top}\mathbf{v}_{t}(\tilde{\mathbf{x}}))\}_{t}$ are stationary and ergodic for any $\mathbf{h} \in \mathbb{R}^{p}$, where $\mathrm{i}=\sqrt{-1}$.
    Thus \eqref{lem1-cond2} follows immediately from the Ergodic theorem. 
    To show \eqref{lem1-cond1}, it suffices to show, with probability one,
    \begin{align*}
        \hat{\varphi}_{n}(\mathbf{h}) := \frac{1}{n-p+1} \sum_{t=p-1}^{n-1} \exp(\mathrm{i}\mathbf{h}^{\top}\mathbf{v}_{t}(\tilde{\mathbf{x}})) \rightarrow \int \exp(\mathrm{i}\mathbf{h}^{\top}\mathbf{v}) d\mu_{p}(\mathbf{v}) =: \varphi(\mathbf{h}) \quad \mbox{for all } \mathbf{h} \in \mathbb{R}^{p}.
    \end{align*}
    By the Ergodic theorem, we have
    \begin{align} \label{lem1-23oct4}
        \hat{\varphi}_{n}(\mathbf{q}) \rightarrow \varphi(\mathbf{q}), \quad \mbox{for all } \mathbf{q} \in \mathbb{Q}^{p}, 
    \end{align}
    with probability one. Let 
    \begin{align*}
        \mathcal{E} = \{\hat{\varphi}_{n}(\mathbf{q}) \rightarrow \varphi(\mathbf{q}) \mbox{ for all } \mathbf{q} \in \mathbb{Q}^{p}\} 
        \cap \left\{ \frac{1}{n-p+1} \sum_{t=p-1}^{n-1}\Vert \mathbf{v}_{t}(\tilde{\mathbf{x}}) \Vert \rightarrow \mathbb{E}\Vert \mathbf{v}_{p-1}(\tilde{\mathbf{x}}) \Vert\right\},
    \end{align*}
    noting that $\mathbb{E} \Vert \mathbf{v}_{p-1}(\tilde{\mathbf{x}}) \Vert < \infty$ by (A2).
    By the Ergodic theorem and \eqref{lem1-23oct4}, $\mathbb{P}(\mathcal{E}) = 1$.
    Fix $\omega \in \mathcal{E}$ and let $\epsilon>0$ and $\mathbf{h} \in \mathbb{R}^{p}$ be arbitrary.
    Since $\varphi$ is uniformly continuous, there exists $\delta>0$ such that 
    \begin{align} \label{lem1-oct1024-1}
        |\varphi(\mathbf{h}) - \varphi(\mathbf{q})| < \epsilon,
    \end{align}
    for all $\mathbf{q} \in \mathbb{Q}^{p}$ with $\Vert \mathbf{h} - \mathbf{q} \Vert < \delta$. 
    In addition, using $|\exp(\mathrm{i}\mathbf{h}^{\top}\mathbf{z}) - 1| \leq \Vert \mathbf{h} \Vert \Vert \mathbf{z} \Vert$,
    we have
    \begin{align} \label{lem1-oct1024-2}
        |\hat{\varphi}_{n}(\mathbf{h}) - \hat{\varphi}_{n}(\mathbf{q})| \leq \frac{\delta}{n-p+1} \sum_{t=p-1}^{n-1} \Vert \mathbf{v}_{t}(\tilde{\mathbf{x}}) \Vert,
    \end{align}
    for all $\mathbf{h}$, $\mathbf{q}$ such that $\Vert \mathbf{h} - \mathbf{q} \Vert < \delta$.
    Combining \eqref{lem1-oct1024-1}--\eqref{lem1-oct1024-2}, we have
    \begin{align*}
        |\hat{\varphi}_{n}(\mathbf{h}) - \varphi(\mathbf{h})| \leq& |\hat{\varphi}_{n}(\mathbf{h}) - \hat{\varphi}(\mathbf{q})| + |\hat{\varphi}_{n}(\mathbf{q}) - \varphi(\mathbf{q})| + |\varphi(\mathbf{q}) - \varphi(\mathbf{h})| \\
        \leq& \frac{\delta}{n-p+1} \sum_{t=p-1}^{n-1} \Vert \mathbf{v}_{t}(\tilde{\mathbf{x}}) \Vert + o(1) + \epsilon
    \end{align*}
    on $\omega \in \mathcal{E}$.
    Note, in particular, that $\epsilon$ and $\delta$ do not depend on $\mathbf{h}$. 
    Letting $\epsilon \rightarrow 0$ and $\delta \rightarrow 0$ proves \eqref{lem1-cond1} holds on $\mathcal{E}$, and hence with probability one.
\end{proof}

\begin{proof}[Proof of Proposition \textcolor{red}{4}] 
    We only consider the case $\mathcal{M}_{n} \subseteq \{n_{1} + 1, \ldots, n-1\}$ here as the other case follows from a similar argument.
    Define
    \begin{align*}
        \hat{\mu}_{\mathrm{pre}} &:= \frac{1}{n_{1}-p+2} \sum_{t=p-1}^{n_{1}} \delta_{\mathbf{v}_{t}(\hat{\mathbf{w}})} = \frac{1}{n_{1}-p+2} \sum_{t=p-1}^{n_{1}} \delta_{\mathbf{v}_{t}(\tilde{\mathbf{x}})} =: \mu_{\mathrm{pre}},\\
        \hat{\mu}_{\mathrm{post}} &:= \frac{1}{n - n_{1} - 1} \sum_{t=n_{1}+1}^{n-1} \delta_{\mathbf{v}_{t}(\hat{\mathbf{w}})} \\
        \mu_{\mathrm{post}} &:= \frac{1}{n - n_{1} - 1} \sum_{t=n_{1}+1}^{n-1} \delta_{\mathbf{v}_{t}(\tilde{\mathbf{x}})}
    \end{align*}
    Since the Wasserstein distance of order $k$ is a metric on the space of probability measures with finite $k$-th moments, we have
    \begin{align*}
        \mathcal{W}_{k}(\hat{\mu}_{\mathrm{post}}, \mu_{p}) \leq& \mathcal{W}_{k}(\hat{\mu}_{\mathrm{post}}, \hat{\mu}_{\mathrm{pre}}) + \mathcal{W}_{k}(\mu_{\mathrm{pre}}, \mu_{p}) \\
        \leq& \mathcal{W}_{k}({\mu}_{\mathrm{post}}, {\mu}_{\mathrm{pre}}) + \mathcal{W}_{k}(\mu_{\mathrm{pre}}, \mu_{p}) \\
        \leq& \mathcal{W}_{k}(\mu_{\mathrm{post}}, \mu_{p})+ 2\mathcal{W}_{k}(\mu_{\mathrm{pre}}, \mu_{p}).
    \end{align*}
    Now the desired result follows from Lemma \ref{lem1}.
\end{proof}

In the rest of this section, the derivation of (\textcolor{red}{13}) in Section \textcolor{red}{4.2} is provided.

\begin{proof}[Proof of (\textcolor{red}{13}) in Section \textcolor{red}{4.2}]
    From (\textcolor{red}{9})--(\textcolor{red}{11}) in the main text, we can compute the marginal distributions implied by the imputed series, given as follows. 
If $t < n_{1}$,
\begin{align*}    
    \mathbb{P}((\hat{w}_{t}, \hat{w}_{t-1}) = (1,1)) =& (1-\frac{1}{k_{1}})\lambda_{2}(1-q) + \frac{1}{k_{1}} \lambda_{2} a_{1}, \\
    \mathbb{P}((\hat{w}_{t}, \hat{w}_{t-1}) = (1,0)) =& (1 - \frac{1}{k_{1}}) \lambda_{1} p + \frac{1}{k_{1}} \lambda_{1} b_{1}, \\
    \mathbb{P}((\hat{w}_{t}, \hat{w}_{t-1}) = (0,1)) =& (1 - \frac{1}{k_{1}}) \lambda_{2} q + \frac{1}{k_{1}} \lambda_{2} (1 - a_{1}), \\
    \mathbb{P}((\hat{w}_{t}, \hat{w}_{t-1}) = (0,0)) =& (1 - \frac{1}{k_{1}}) \lambda_{1} (1 - p) + \frac{1}{k_{1}} \lambda_{1} (1 - b_{1}),
\end{align*}
and for $t \geq n_{1}$, the marginal distribution can be obtained by replacing $a_{1}, b_{1}, k_{1}$ by $a_{2}, b_{2}, k_{2}$ in the above. 
Since TWI equates the marginal distributions before and after $n_{1}$, we have 
\begin{align*}
    (1-\frac{1}{k_{1}})\lambda_{2}(1-q) + \frac{1}{k_{1}} \lambda_{2} a_{1} =& (1-\frac{1}{k_{2}})\lambda_{2}(1-q) + \frac{1}{k_{2}} \lambda_{2} a_{2}, \\
    (1 - \frac{1}{k_{1}}) \lambda_{1} p + \frac{1}{k_{1}} \lambda_{1} b_{1} =& (1 - \frac{1}{k_{2}}) \lambda_{1} p + \frac{1}{k_{2}} \lambda_{1} b_{2}, \\
    (1 - \frac{1}{k_{1}}) \lambda_{2} q + \frac{1}{k_{1}} \lambda_{2} (1 - a_{1}) =& (1 - \frac{1}{k_{2}}) \lambda_{2} q + \frac{1}{k_{2}} \lambda_{2} (1 - a_{2}), \\
    (1 - \frac{1}{k_{1}}) \lambda_{1} (1 - p) + \frac{1}{k_{1}} \lambda_{1} (1 - b_{1}) =& (1 - \frac{1}{k_{2}}) \lambda_{1} (1 - p) + \frac{1}{k_{2}} \lambda_{1} (1 - b_{2}).
\end{align*}
After some algebraic manipulations, these equations simplify to (\textcolor{red}{13}) in the main text.
\end{proof}

\section{Scalar filter algorithm} \label{App::SF}
The following summarizes the scalar filter algorithm of \citet{pena2021}.
\begin{itemize}
    \item Initialize the imputation $\hat{\mathbf{w}}^{(0)} = (\hat{w}_{0}^{(0)}, \hat{w}_{1}^{(0)}, \ldots, \hat{w}_{n-1}^{(0)})^{\top}$ with $\hat{w}_{t}^{(0)} = 0$ if $t \in \mathcal{M}$.
    \item For each $s \in \mathcal{M}$, let $I^{(s)}_{t}$ be the indicator function. That is,
    \begin{align*}
        I_{t}^{(s)} = \left\{ \begin{array}{cc}
           1, & t = s \\
           0,  & t \neq s
        \end{array} \right.
    \end{align*}
    and estimate the intervention model using $\hat{\mathbf{w}}^{(0)}$. For example, one can fit an ARX model with $\{I_{t}^{(s)}: s\in \mathcal{M}\}$ as exogenous predictors.
    \item Let $\hat{\beta}^{(s)}$ be the intervention effect associated with $I_{t}^{(s)}$. Estimate the missing data by
    \begin{align*}
        \hat{w}_{s} = \hat{w}_{s}^{(0)} - \hat{\beta}^{(s)}, \quad s\in \mathcal{M}.
    \end{align*}
\end{itemize}
As one can see from the algorithm above, the scalar filter algorithm proposed by \citet{pena2021} is based on intervention analysis.
It can be iterated in order to improve the estimates.
However, no convergence guarantee is provided by the authors.

When fitting an ARX model with the indicators as exogenous predictors, if there are many missing values, standard approaches may be infeasible or unreliable due to near collinearity. 
In this case, we employ a small $\ell_{2}$ penalty to avoid numerical instability. 

\section{Complementary numerical results} \label{app::ACFs}

\subsection{Autocovariance estimation}

In this subsection, we compare the ACFs estimated from the imputed series.
Define $\gamma(\mathbf{z}, h) = (n-h)^{-1} \sum_{t=h}^{n-1}(z_{t} - \bar{z})(z_{t-h} - \bar{z})$, where $\mathbf{z} = (z_{0}, z_{1}, \ldots, z_{n-1})$ and $\bar{z} = n^{-1}\sum_{t=0}^{n-1}z_{t}$.
Then the RMSE is defined as
\begin{align} \label{Sec5_acf_def}
    \mathrm{RMSE}_{\gamma}(h) = \sqrt{\frac{1}{1000} \sum_{i=1}^{1000} (\gamma(\hat{\mathbf{w}}^{(i)}, h) - \gamma_{\mathbf{x}}(h))^{2}},
\end{align}
where $\hat{\mathbf{w}}^{(i)} = (\hat{w}_{0}^{(i)}, \ldots, \hat{w}_{n-1}^{(i)})$ and $\gamma_{\mathbf{x}}(h) = \frac{1}{1000} \sum_{i=1}^{1000} \gamma(\mathbf{x}^{(i)}, h)$ with $\mathbf{x}^{(i)} = (\tilde{x}_{0}^{(i)}, \ldots, \tilde{x}_{n-1}^{(i)})$.

Table \ref{tab:acf_RMSE} reports $\mathrm{RMSE}_{\gamma}(h)$ for $h = 0, 1,2$. 
For Model \textcolor{red}{1} 
(AR), most methods perform reasonably well. 
Linear interpolation produces accurate ACF estimates since the positive AR coefficient aligns well with linearity, while the scalar filter exhibits largest estimation errors.
For Model \textcolor{red}{2} 
(ARMA), TWI improves from linear and Kalman initializations and yields more accurate ACF estimates: TWI reduces the biases in Kalman smoothing when $h = 0$ and in linear interpolation when $h = 1$, and $h = 2$.
For the nonlinear Model \textcolor{red}{3} 
(TAR), the ACFs estimated from the TWI-imputed series have RMSEs that are among the lowest, with substantial improvements over the initializations.
Notably, the scalar filter also yields accurate ACFs for $h = 1, 2$. 
Figure \ref{fig:TAR_acfhist} plots the histogram of the estimated ACFs at lag 1 (that is, $\gamma(\hat{\mathbf{w}}^{(i)}, 1)$), which shows how TWI effectively nullifies the biases caused by the initializations. 

For DGP \textcolor{red}{4} (I(1)) and DGP \textcolor{red}{5} (CYC), TWI again offers great improvement over initializations.
Noe that for DGP \textcolor{red}{5} (CYC), which generates a structural time series with deterministic cycles, iSSA is very effective at least under missing pattern I.
However, perhaps due to the small window length caused by the high missing ratio, the performance of iSSA under missing pattern II is severely compromised (see also Figure \ref{fig:Cyc_scatter}, which shows TWI captures the underlying dynamic well).
The scalar filter performs poorly for these models, which is likely the consequence of using many indicators in the intervention analysis.
Among the benchmarks, Kalman smoothing produces the best result for these two models, which can benefit from further improvement provided by the proposed TWI procedure.


\spacingset{1}
\begin{figure}
    \centering
    \begin{subfigure}[t]{0.48\textwidth}
        \centering
        \includegraphics[width=\linewidth, height = 0.35\textheight]{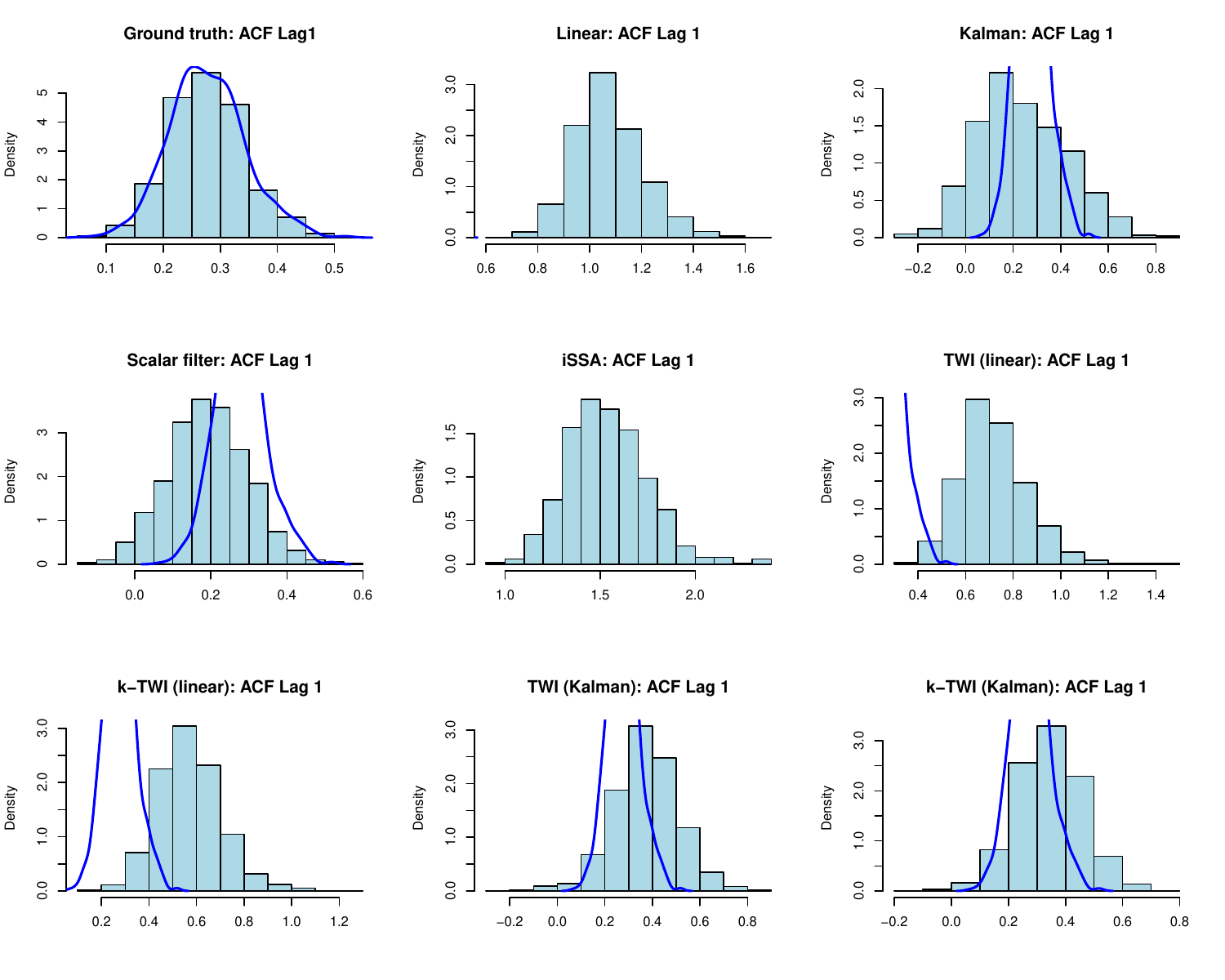}
        \caption{Missing pattern I}
    \end{subfigure} 
    \begin{subfigure}[t]{0.48\textwidth}
        \centering
        \includegraphics[width=\linewidth, height = 0.35\textheight]{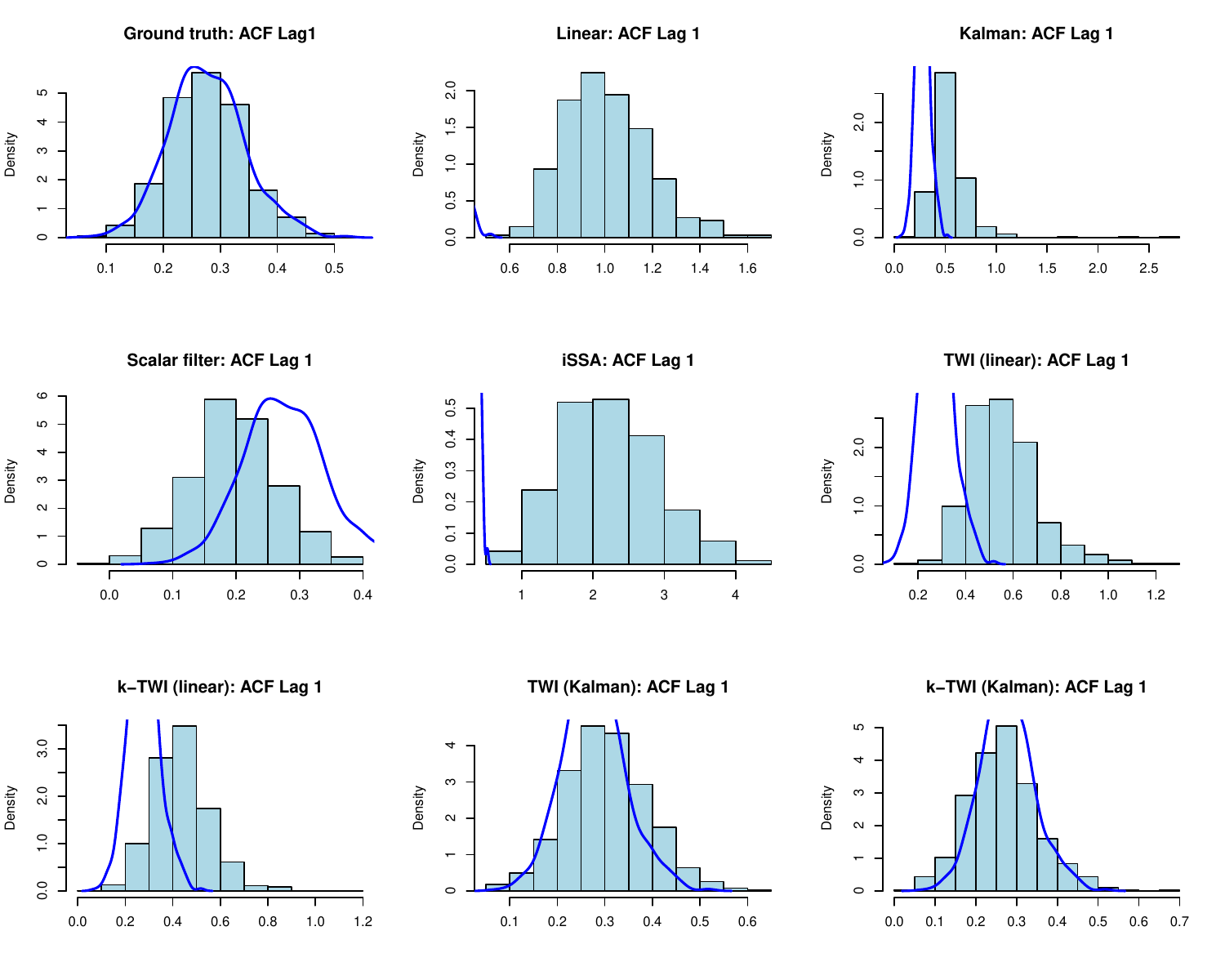}
        \caption{Missing pattern II}
    \end{subfigure}
    \caption{Histogram of estimated autocovariance function at lag 1 when the data are generated from Model \textcolor{red}{3} 
    (TAR). The blue line is the smoothed density of the histogram of the ACFs estimated from the ground truth.}
    \label{fig:TAR_acfhist}
\end{figure}
\spacingset{1.75}

\spacingset{1}
\begin{figure}
    \centering
    \begin{subfigure}[t]{0.49\textwidth}
        \includegraphics[width=\linewidth]{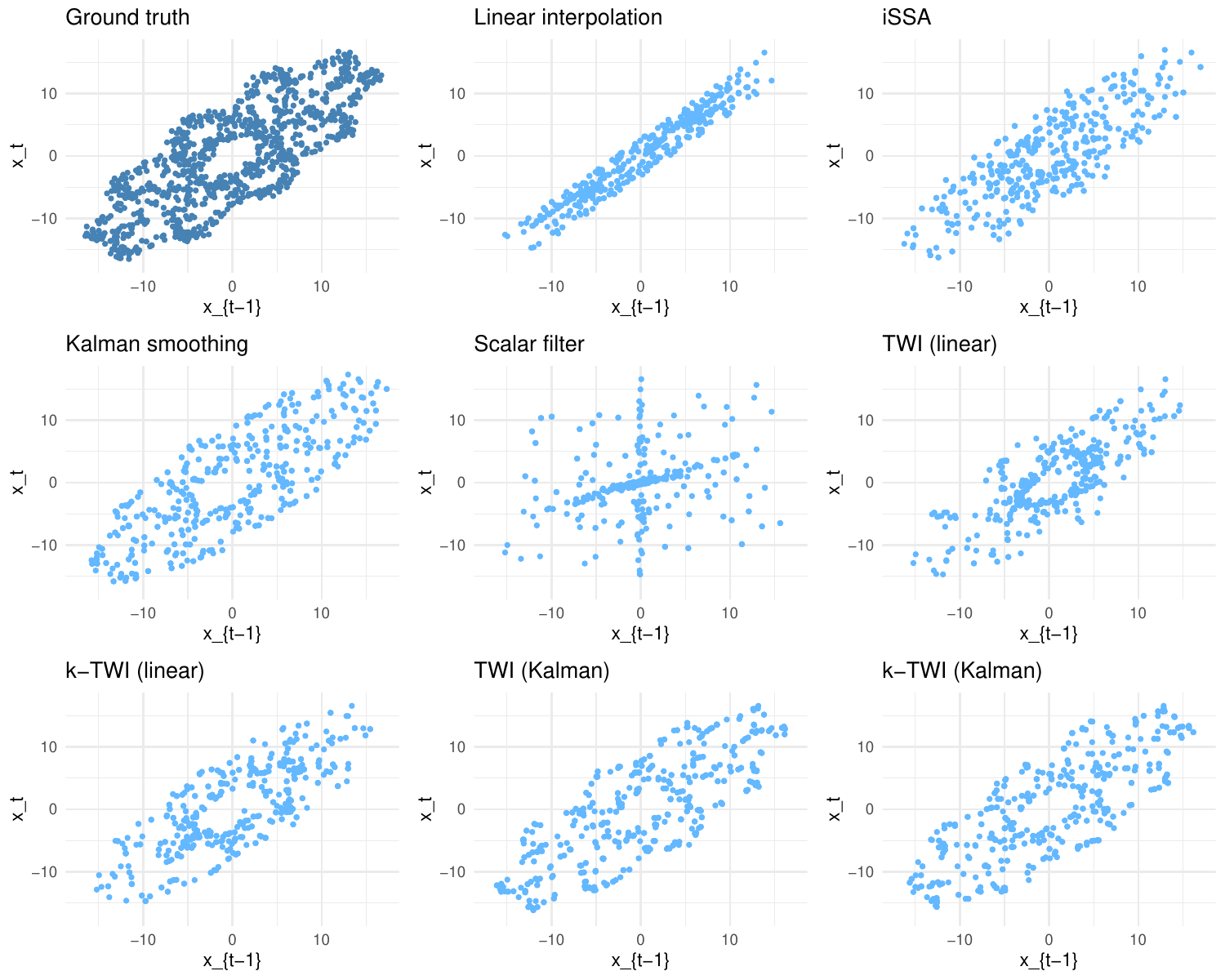}
        \caption{$h = 1$}
    \end{subfigure}
    \begin{subfigure}[t]{0.49\textwidth}
        \includegraphics[width=\linewidth]{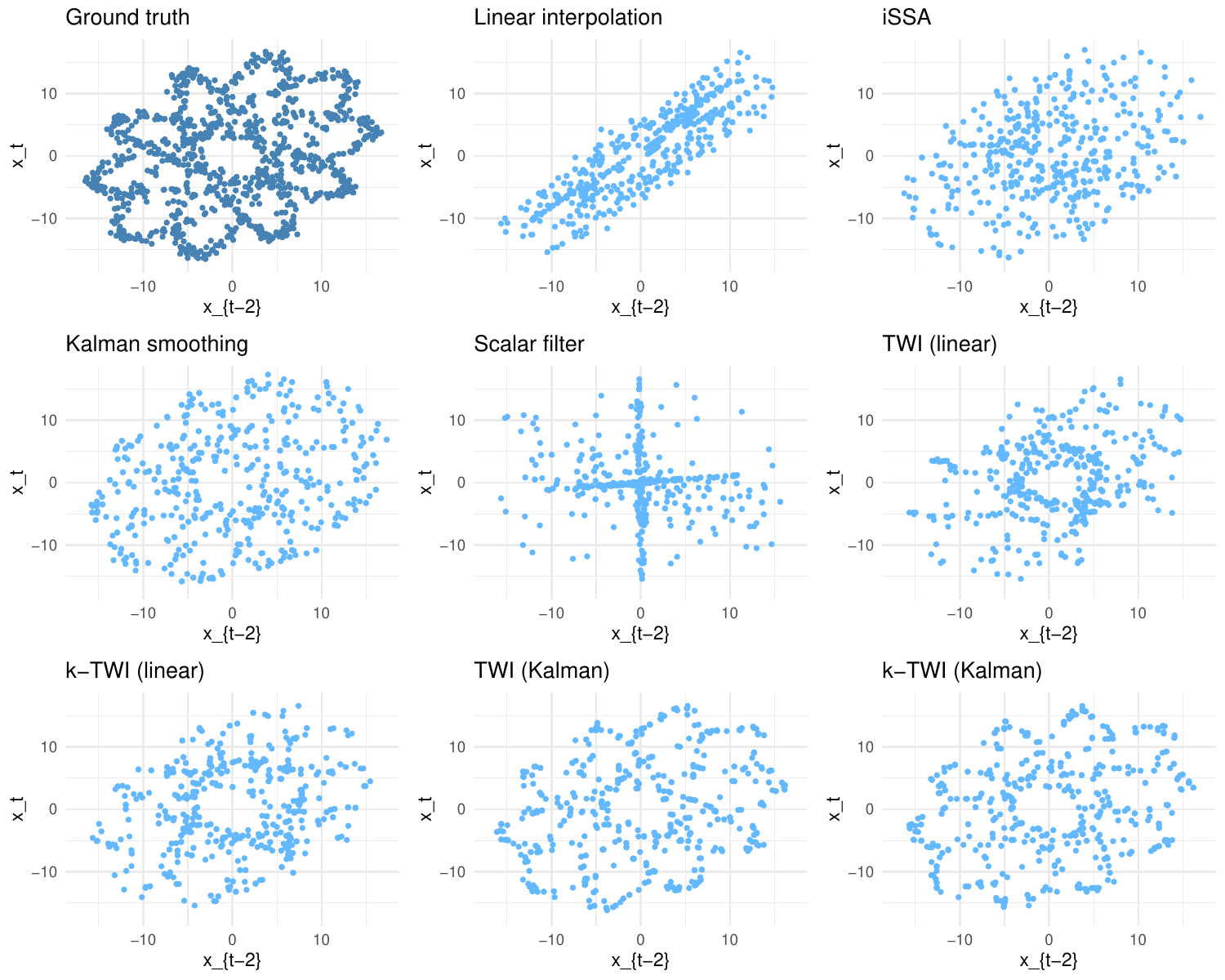}
        \caption{$h = 2$}
    \end{subfigure} \\
    \begin{subfigure}[t]{0.49\textwidth}
        \includegraphics[width=\linewidth]{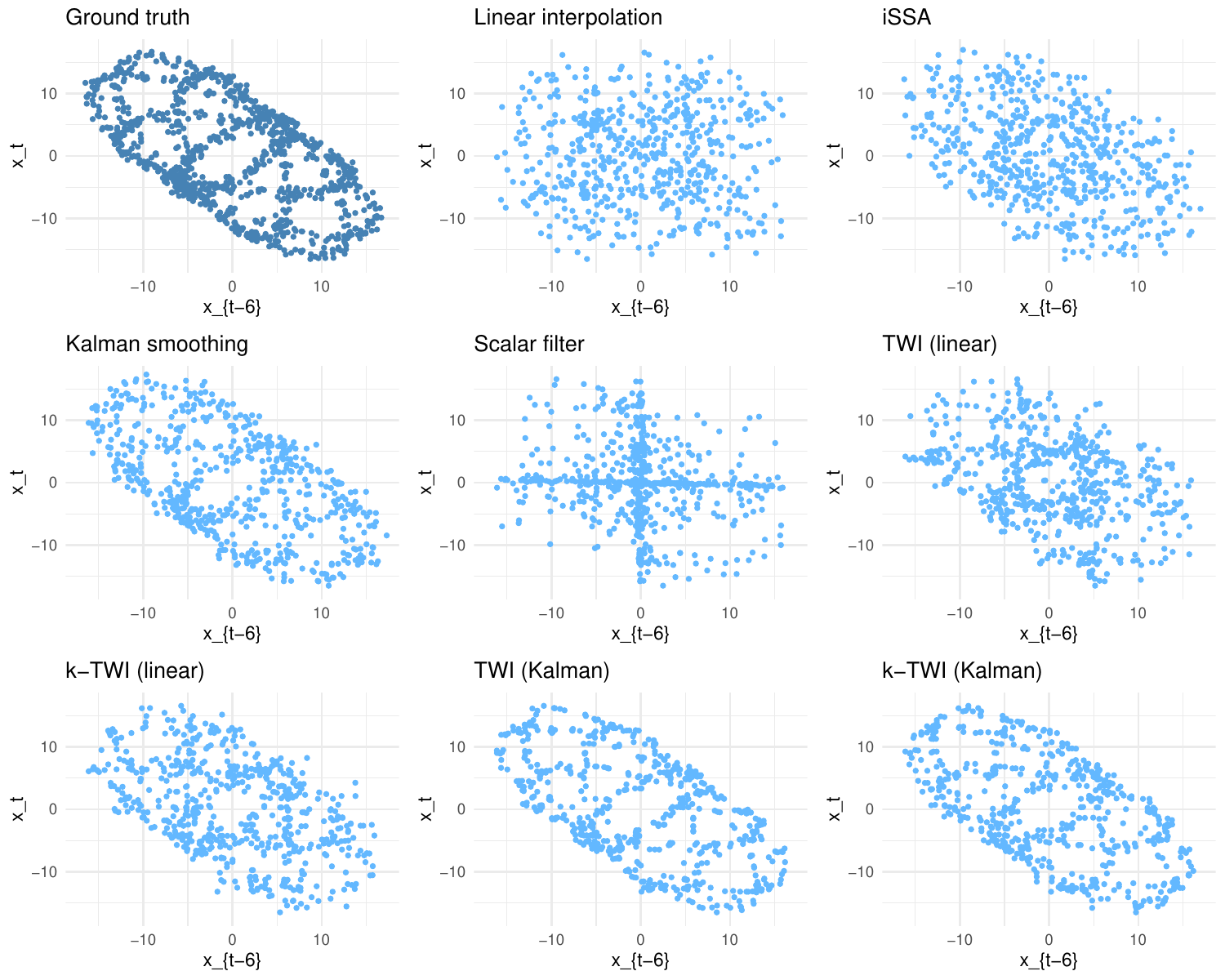}
        \caption{$h = 6$}
    \end{subfigure}
    \begin{subfigure}[t]{0.49\textwidth}
        \includegraphics[width=\linewidth]{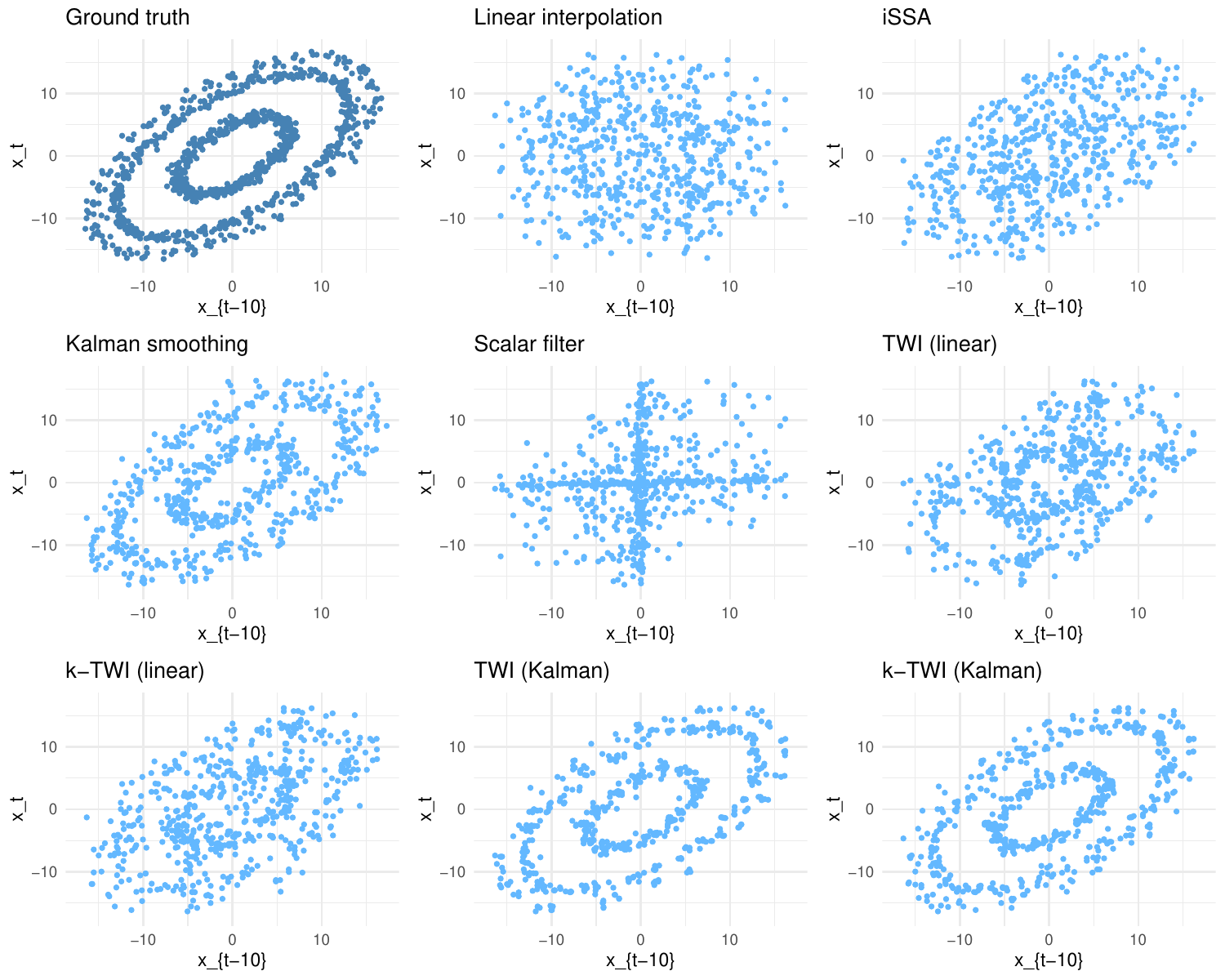}
        \caption{$h = 10$}
    \end{subfigure}
    \caption{Scatter plots of $(\tilde{x}_{t-h}, \tilde{x}_{t})$ of the original data and $(\hat{w}_{t-h}, \hat{w}_{t})$ of the imputed series, when the data are generated from Model \textcolor{red}{5} (CYC) under missing pattern II. $h \in \{1,2,6,10\}$.}
    \label{fig:Cyc_scatter}
\end{figure}
\spacingset{1.75}

\spacingset{1}
\begin{table}[h!]
\centering
\caption{$\mathrm{RMSE}_{\gamma}(h)$, defined in \eqref{Sec5_acf_def}, for $h = 0, 1, 2$, with $h$ denoting lag. The results are averaged over 1000 simulations, and the smallest values in each row are marked in boldface. See Table \textcolor{red}{1} 
in the main text for notations.}
\begin{tabular}{@{}ccrrrrrrrr@{}}
\toprule
Model& $h$ & Linear & iSSA & Kalman & ScalarF & TWI$_\mathrm{lin}$ & $k$-TWI$_\mathrm{lin}$ & TWI$_\mathrm{Kal}$ & $k$-TWI$_\mathrm{Kal}$        \\ \midrule
\multicolumn{10}{c}{\bf Missing pattern I} \\
AR    & 0 &  0.32 & {\bf 0.28} & 0.37 & 0.50 & 0.30 & 0.31 & 0.32 & 0.32 \\
      & 1 &  {\bf 0.26} & 0.31 & 0.29 & 0.47 & {\bf 0.26} & 0.30 & 0.28 & 0.31 \\
      & 2 &  0.26 & 0.28 & 0.29 & 0.38 & {\bf 0.25} & 0.28 & 0.27 & 0.29\smallskip \\
ARMA  & 0 &  0.14 & {\bf 0.09} & 0.29 & 0.31 & 0.13 & 0.13 & 0.22 & 0.18 \\
      & 1 &  0.22 & 0.34 & {\bf 0.07} & 0.08 & 0.14 & 0.10 & {\bf 0.07} & {\bf 0.07} \\
      & 2 &  0.08 & 0.11 & 0.07 & 0.07 & {\bf 0.06} & {\bf 0.06} & 0.07 & 0.07\smallskip \\
TAR   & 0 &  0.59 & {\bf 0.35} & 0.73 & 1.07 & 0.65 & 0.56 & 0.56 & 0.45 \\
      & 1 &  0.81 & 1.28 & 0.19 & 0.14 & 0.46 & 0.32 & 0.17 & {\bf 0.13} \\
      & 2 &  0.40 & 0.40 & 0.16 & 0.11 & 0.24 & 0.16 & 0.10 & {\bf 0.09}\smallskip \\
I(1)  & 0 &  0.52 & 0.48 & 0.48 & 19.41 & 0.42 & 0.35 & 0.40 & {\bf 0.34} \\
      & 1 &  0.35 & 0.40 & 0.25 & 9.32 & 0.28 & 0.25 & 0.22 & {\bf 0.21} \\
      & 2 &  0.17 & 0.13 & {\bf 0.11} & 1.28 & 0.15 & 0.14 & {\bf 0.11} & {\bf 0.11}\smallskip \\
CYC   & 0 &  10.94 & 0.64 & 3.90 & 4.67 & 1.81 & 1.25 & 0.68 & {\bf 0.62} \\
      & 1 &  9.05 & {\bf 0.43} & 3.33 & 5.88 & 1.47 & 1.04 & 0.54 & 0.51 \\
      & 2 &  1.90 & {\bf 0.15} & 0.81 & 2.13 & 0.41 & 0.33 & 0.17 & 0.18\smallskip   \\
\multicolumn{10}{c}{\bf Missing pattern II} \\
AR    & 0 &  0.38 & 0.34 & 0.55 & 0.78 & {\bf 0.33} & 0.36 & 0.43 & 0.39 \\
      & 1 &  {\bf 0.28} & 0.42 & 0.42 & 0.70 & 0.29 & 0.34 & 0.37 & 0.36  \\
      & 2 &  {\bf 0.26} & 0.39 & 0.36 & 0.65 & 0.27 & 0.31 & 0.34 & 0.34\smallskip \\
ARMA  & 0 &  0.14 & 0.44 & 0.31 & 0.32 & 0.13 & {\bf 0.11} & 0.24 & 0.17 \\
      & 1 &  0.17 & 0.65 & 0.08 & 0.09 & 0.10 & 0.08 & {\bf 0.07} & {\bf 0.07} \\
      & 2 &  0.15 & 0.46 & {\bf 0.07} & 0.08 & 0.09 & {\bf 0.07} & {\bf 0.07} & {\bf 0.07}\smallskip \\
TAR   & 0 &  {\bf 0.49} & 1.29 & 0.71 & 1.08 & 0.72 & 0.68 & 0.73 & 0.65 \\
      & 1 &  0.75 & 2.06 & 0.34 & 0.11 & 0.31 & 0.20 & {\bf 0.09} & {\bf 0.09} \\
      & 2 &  0.77 & 1.37 & 0.34 & {\bf 0.10} & 0.42 & 0.28 & 0.19 & 0.14\smallskip \\
I(1)  & 0 &  0.47 & 0.38 & 0.47 & 51.84 & 0.38 & 0.30 & 0.36 & {\bf 0.29} \\
      & 1 &  0.19 & 0.30 & 0.18 & 0.56 & 0.17 & 0.14 & 0.14 & {\bf 0.13} \\
      & 2 &  0.08 & {\bf 0.05} & 0.07 & 0.52 & 0.08 & 0.08 & 0.07 & 0.07\smallskip \\
CYC   & 0 &  7.46 & 7.04 & 1.96 & 15.21 & 12.27 & 6.49 & 1.26 & {\bf 1.03} \\
      & 1 &  2.63 & 4.82 & 1.45 & 16.23 & 9.57 & 5.07 & 0.98 & {\bf 0.83} \\
      & 2 &  11.18 & 0.70 & 0.79 & 7.90 & 1.28 & 1.04 & 0.28 & {\bf 0.31} \\
\bottomrule
\end{tabular}
\label{tab:acf_RMSE}
\end{table}
\spacingset{1.75}

\subsection{Compositional time series}
In this subsection, we discuss the simulation results associated with the following model.

\setcounter{model}{6}
\begin{model}[Additive logistic model; AL, \citealp{brunsdon1998}] \label{Sec5_modelAL}
\begin{align*}
    y_{t,1} =& 0.1 + 0.7 y_{t-1,1} - 0.5 y_{t-1,2} + 0.2 \epsilon_{t,1} \\
    y_{t,2} =& 0.1 - 0.7 y_{t-1,2} + 0.2 \epsilon_{t,2}
\end{align*}
and the observed data are
\begin{equation*}
    x_{t,j} = \left\{ \begin{aligned}
        &\frac{\exp(y_{t,j})}{1 + \exp(y_{t,1}) + \exp(y_{t,2})}, \quad j = 1, 2, \\
        &\frac{1}{1 + \exp(y_{t,1}) + \exp(y_{t,2})}, \quad j = 3.
    \end{aligned} \right. 
\end{equation*}
\end{model}
DGP \ref{Sec5_modelAL} generates the so-called compositional time series, which refers to a multivariate time series $\mathbf{x}_{t} = (x_{t,1}, \ldots, x_{t,d})^{\top}$, $d > 1$, such that for each $t$, $\sum_{j=1}^{d}x_{t,j} = 1$, and $x_{t,j} \geq 0$ for all $j$.
Such data are encountered in a wide array of fields such as the environmental sciences, economics, geology and political science.
Here, we adopt the additive logistic (AL) model of \citet{brunsdon1998} in DGP \ref{Sec5_modelAL}.
The AL model postulates that after some nonlinear transformation, the data follow a standard VARMA model.
Alternative models for composition time series include the recent works of \citet{ZHENG2017}, \citet{HARVEY2024}, and \citet{Zhu2024}.

Table \textcolor{red}{1} 
in the main paper records the Wasserstein distance between the 3-dimensional marginal distributions of the imputed and original series.
With the smallest Wasserstein loss, TWI methods again demonstrate their ability to adapt to nonlinear dependence structures.
It is worth noting that while the linear interpolation results in a compositional series, other methods, such as the Kalman smoothing, do not generally maintain this compositional property.
By restricting the admissible set, TWI ensures that the imputed series remain compositional. 
As shown in Figure \ref{fig:ALR_scatter}, linear interpolation falsely promotes negative correlation between $\hat{w}_{t,3}$ and $\hat{w}_{t-1,2}$ and Kalman smoothing often deviates from the plausible range of a compositional series.
TWI amends these peculiarities, leading to imputations that align more closely with the underlying data. 

\spacingset{1}
\begin{figure}
    \centering
    \begin{subfigure}[t]{0.48\textwidth}
    \centering
        \includegraphics[width=\linewidth, height=0.4\textheight]{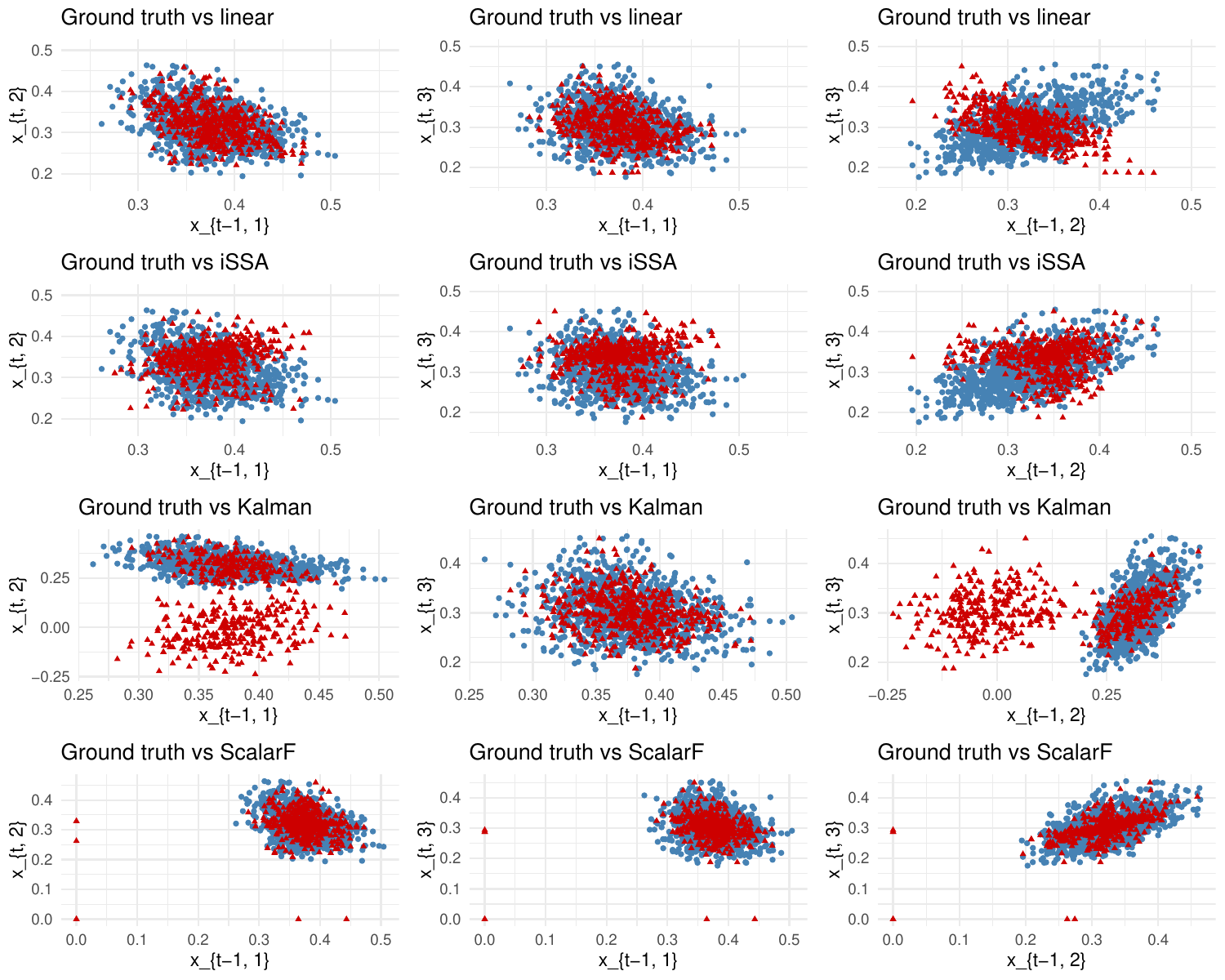}
    \end{subfigure} 
    \begin{subfigure}[t]{0.48\textwidth}
    \centering
        \includegraphics[width=\linewidth, height=0.4\textheight]{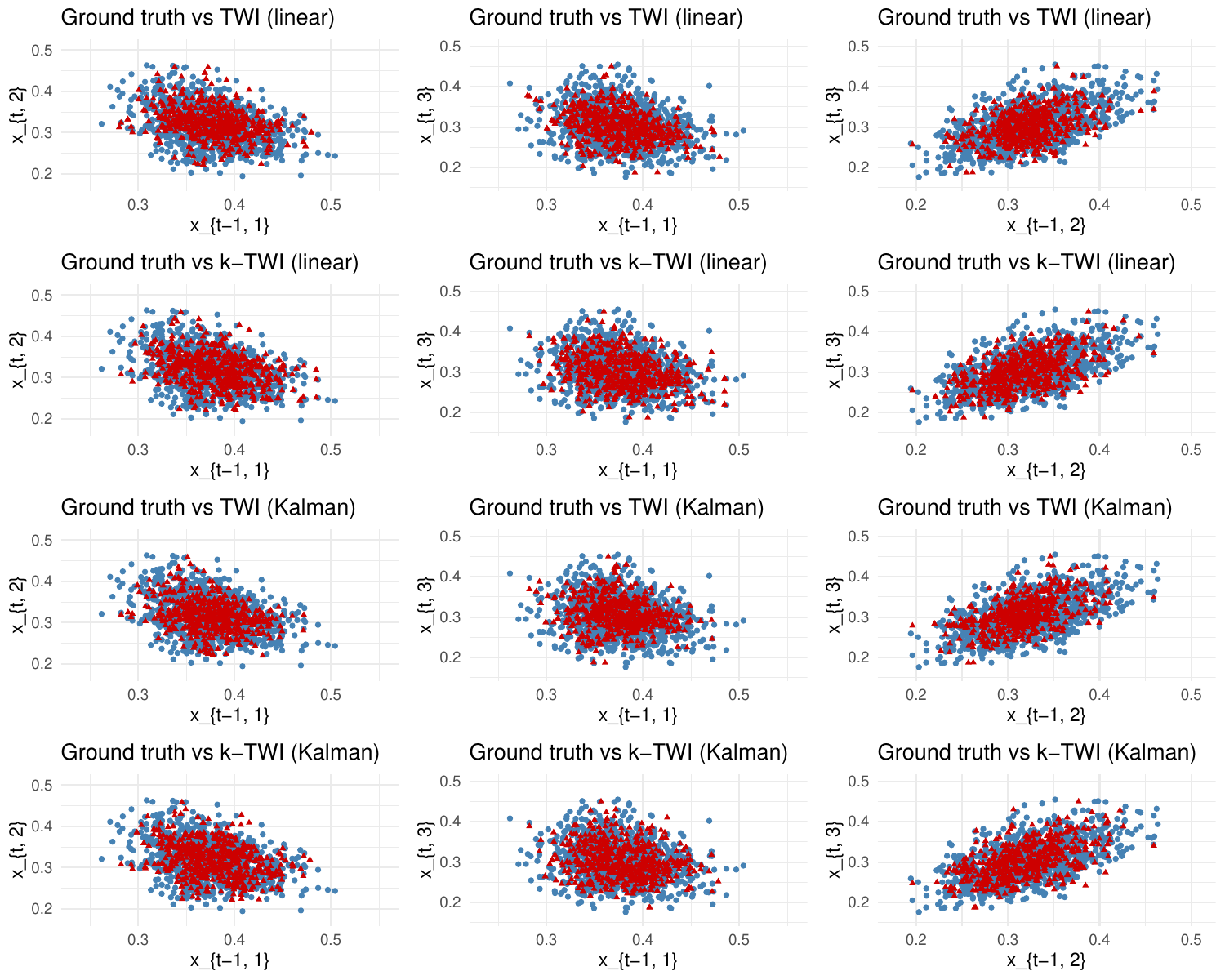}
    \end{subfigure} 
    \caption{Scatter plots of $(x_{t-1,i}, x_{t,j})$ of the original data (plotted in blue) and $(\hat{w}_{t-1,i}, \hat{w}_{t,j})$ of the imputed series (plotted in red) for $i,j = 1,2, 3$, $i \neq j$, when the data are generated from Model \ref{Sec5_modelAL} (AL) under missing pattern I.}
    \label{fig:ALR_scatter}
\end{figure}
\spacingset{1.75}

\subsection{Implementation details in groundwater data analysis}

The methods employed in the groundwater data analysis are detailed as follows.
First, linear interpolation, mean imputation, and Kalman smoothing are implemented using the \texttt{imputeTS} package \citep{imputets} in \texttt{R}.
For the iterative singular spectrum analysis (iSSA) method, we implement it via the \texttt{Rssa} package \citep{golyandina2018}.
There are three parameters to tune for the iSSA: the initialization, the window length, and the rank in the reconstruction step.
Mean imputation is adopted here as the initialization, which is a common choice in practice.
We set the window length to 6, which is roughly the maximum window length allowed when the missing ratio is high.
Finally, to determine the rank in the reconstruction step, we employ a cross-validation approach.
Specifically, we randomly mask 10\% of the observed data. Then, for each rank $r \in \{1,2,\ldots,6\}$, we apply the iSSA for imputation and select the rank $r^{*}$ that yields the most accurate imputations on the masked data, measured by the mean squared error.
Finally, the iSSA is implemented again with the full sample and the selected rank.

For the TWI-type methods, as discussed in the main text, we first apply kernel smoothing to estimate a trend function $f(t)$. We use the \texttt{locfit} package \citep{locfit} with the default tricube kernel and the bandwidth is set to 12 months. This step can be implemented regardless of the missing pattern.
Then the ($k$-)TWI is applied to the residuals $r_{t} = x_{t} - f(t)$, which has the same missing pattern as $\{x_{t}\}$, to obtain $\{\hat{r}_{t}\}$.
For the TWI, the cut-off point $n_{1}$ is set to $\lfloor 0.4 n \rfloor$ as in the simulation section.
For $k$-TWI, the cut-offs are $\{0.3n, 0.7n, 0.5n\}$.
The simple last observation carried forward (LOCF) is used as initialization for both TWI and $k$-TWI.
Finally, the imputation is $\hat{x}_{t} = \hat{r}_{t} + f(t)$.

\end{document}